\documentclass[11pt,letterpaper]{article}
\usepackage{amsmath,amsfonts,amsthm, mathrsfs,mathtools}
\usepackage{dsfont}
\usepackage{geometry}
\numberwithin{equation}{section}   
\usepackage{booktabs} 
\usepackage[ruled,noend]{algorithm2e} 
\usepackage{multirow}
\usepackage{natbib}
\usepackage{tikz}
\usepackage{pgfplots}
\pgfplotsset{compat=1.18}

\SetAlFnt{\small}
\SetAlCapFnt{\small}
\SetAlCapNameFnt{\small}
\SetAlCapHSkip{0pt}
\IncMargin{-\parindent}

\usepackage[thinc]{esdiff}

\numberwithin{equation}{section}
\usepackage{packages/macros-common}
\usepackage{packages/macros-project-specific}
\usepackage{packages/mytheorems}

\usepackage[suppress]{color-edits}
\addauthor{ss}{magenta} 
\addauthor{gp}{blue}

\renewcommand\thefootnote{\textcolor{blue}{\arabic{footnote}}}

\allowdisplaybreaks
\title{Single-Sample Bilateral Trade with a Broker}

\newcommand{\country}[1]{#1.}
\newcommand{\city}[1]{#1}
\newcommand{\institution}[1]{#1}
\newcommand{\email}[1]{Email: \texttt{#1}}
\newcommand{\affiliation}{\thanks}

\newcommand\extrafootertext[1]{%
    \bgroup
    \renewcommand\thefootnote{\fnsymbol{footnote}}%
    \renewcommand\thempfootnote{\fnsymbol{mpfootnote}}%
    \footnotetext[0]{#1}%
    \egroup
}

\author{
 {MohammadTaghi Hajiaghayi
 \affiliation{
   \institution{University of Maryland}
   \city{College Park}
   \country{USA}
 \email{hajiagha@umd.edu}
 }}
 \and
 {Gary Peng
 \affiliation{
   \institution{University of Maryland}
   \city{College Park}
   \country{USA}
 \email{gpeng1@terpmail.umd.edu }
 }}
 \and
 {Suho Shin
 \affiliation{
   \institution{University of Maryland}
   \city{College Park}
   \country{USA}
 \email{suhoshin@umd.edu}
 }}
}

\date{}
\begin{document}

\maketitle

\begin{abstract}

We initiate the study of single-sample bilateral trade with a broker, drawing an analogy to the setting of single-sample bilateral trade \textit{without} a broker considered in~\cite{babaioff2020bulow} and~\cite{cai2023fixed}. Our model captures the three-sided interaction in which a broker mediates trade between a buyer and seller, each described by a valuation distribution from which a single sample can be drawn.

We consider two settings in particular: one where the valuation distributions of the buyer and seller are identical and one where the valuation distributions are stochastically ordered. We analyze simple mechanisms that rely only on a single sample from each agent's distribution and show that these mechanisms achieve constant-factor approximations to the first-best gains-from-trade (GFT), first-best social welfare (SW), and optimal profit under the standard monotone-hazard-rate assumption. We then complement these results with matching or nearly matching upper bounds on the GFT and SW of our mechanisms. Notably, in both settings, we observe fairly small losses in the approximation factors to the first-best GFT and first-best SW due to the existence of the broker (benchmarked against the corresponding approximation factors in the setting without a broker).
Furthermore, our results stand in stark contrast to those of~\cite{hajiaghayi2025bilateral}, who show inapproximability results under a strategic broker with full distributional knowledge.

Our results provide insight into the design of data-efficient brokerage mechanisms for online marketplaces and decentralized trading platforms, where intermediaries must facilitate trade under severe informational constraints. They highlight how even minimal data can enable robust and incentive-compatible brokerage in uncertain markets for both the broker and the market participants.

\extrafootertext{A part of this work has appeared in WWW'26.}
\end{abstract}


\section{Introduction}
Bilateral trade, introduced in the seminal work of~\citet{myerson1983efficient}, is a fundamental problem in mechanism design in which a seller wishes to sell a single item to a buyer. 
Both agents hold private valuations on the item and aim to maximize their quasi-linear utility.
A mechanism designer's objective is to maximize the social welfare induced by the trade, \ie design a mechanism that transfers the item whenever the buyer's valuation is at least the seller's valuation. Such a mechanism is called \textit{ex-post efficient}.
At the same time, a mechanism should satisfy reasonable desiderata to (i) incentivize the buyer and seller to participate (individual rationality) (ii) incentivize the buyer and seller to report their valuations truthfully (incentive compatibility) and (iii) neither subsidize nor profit from the market (budget-balancedness).

Unfortunately,~\citet{myerson1983efficient} prove an impossibility result barring the existence of ex-post efficient mechanisms that are simultaneously incentive-compatible (IC), individually-rational (IR), and budget-balanced (BB).
This pessimistic result holds even when IC is loosened to hold in an ex-ante manner, known as Bayes-Nash incentive compatibility (BNIC), and BB is relaxed to weak budget-balancedness (WBB), \ie we allow the mechanism to profit from the market.
As such, there has been a long line of work studying the \textit{approximability} of the first-best efficiency by such mechanisms, \eg\cite{mcafee2008gains},~\cite{blumrosen2016almost},~\cite{blumrosen2016approximating},~\cite{kang2022fixed},~\cite{deng2022approximately}, and~\cite{cai2023fixed}.

\begin{table*}[t]
    \begin{minipage}{\textwidth}
    \renewcommand{\thempfootnote}{\fnsymbol{mpfootnote}}
    \begin{center}
        \begin{tabular}{|c|c|c|c||c|c|}
             \hline
             \multirow{2}{*} & \multicolumn{3}{c||}{With Broker\footnote{with respect to the mechanism that draws $p, q \sim F$ and posts $\max(p, q)$ to the buyer and $\min(p, q)$ to the seller}} & \multicolumn{2}{c|}{Without Broker\footnote{with respect to the mechanism that draws $p \sim F$ and posts $p$ to both buyer and seller}}\\
             \cline{2-6}
              & GFT & SW & Profit & GFT & SW\\
             \hline
             Lower Bound & $7 / 24$ & $2 / 3$ & $2 / 55$\footnote{if $F$ has a doubly monotone hazard rate} & $1 / 2$ (\cite{babaioff2020bulow}) & $3 / 4$ (\cite{kang2022fixed})\\
             \hline
             Upper Bound & $7 / 24$ & $2 / 3$ & --- & $1 / 2$ (\cite{babaioff2020bulow}) & $3 / 4$ (\cite{kang2022fixed})\\
             \hline
        \end{tabular}
    \end{center}
    \end{minipage}
    \caption{Our results for the symmetric setting.}
    \label{tab:sym-results}
\end{table*}

On the other hand, numerous real-world trades, particularly those in online marketplaces and decentralized trading platforms, involve two parties exchanging an item through a \emph{broker}.
Importantly, the broker may have objectives misaligned with those of society, \eg maximizing her own profit.
In spite of numerous works studying approximate efficiency in bilateral trade \textit{without} a broker (or equivalently, with a broker who seeks to maximize social welfare), there had been few works studying approximate efficiency in bilateral trade with a \textit{strategic} broker until~\citet{hajiaghayi2025bilateral} studied bilateral trade under a profit-maximizing broker who knows the valuation distributions of the buyer and seller.

In practice, however, the broker typically does not know the agents' distributions \textit{a priori}, but may instead only possess partial information based on previous trades, \eg a series of bid samples each agent has submitted.
Motivated by such scenarios, we study approximate efficiency in bilateral trade where the broker only has access to a single sample from each agent's distribution, a well-studied model in the literature on mechanism design under limited information. We study two settings in particular: (i) the \emph{symmetric} setting where the distributions of the buyer and seller are identical\footnote{The symmetric setting is well-motivated in practice, \eg stock markets, treasury markets, and dark pools. It is also well-studied in the literature, \eg~\cite{kang2019fixed},~\cite{babaioff2020bulow}, and~\cite{kang2022fixed}, particularly when the mechanism designer only has sample access to the distributions.} and (ii) the setting where the buyer's distribution \textit{stochastically dominates} the seller's distribution.\footnote{The stochastic dominance setting is widely adopted in the literature and is in fact one of the most standard assumptions for asymmetric distributions when the mechanism designer has limited information, \eg~\cite{babaioff2020bulow},~\cite{maskin2000asymmetric}, and~\cite{kirkegaard2012mechanism}.} In each setting, we show that a simple posted-pricing mechanism achieves a constant approximation to both the first-best gains-from-trade and first-best social welfare. Furthermore, when the distributions of the agents have monotone hazard rates,\footnote{MHR distributions are widely adopted in single-sample mechanism design, \eg~\cite{dhangwatnotai2010revenue}, and more broadly in the mechanism design literature, \eg~\cite{myerson1981optimal},~\cite{hartline2009simple}, and~\cite{blumrosen2016approximating}.
In particular,~\citet{myerson1983efficient} characterize the profit-maximizing BNIC and IR mechanism for bilateral trade under regular (and hence MHR) problem instances, when the broker knows the distributions of the buyer and seller.} we show that our mechanisms achieve a constant approximation to the optimal profit attained by any BNIC and IR mechanism with full distributional knowledge (henceforth referred to simply as the optimal profit). A summary of our results can be found in Tables~\ref{tab:sym-results} and~\ref{tab:asym-results}.

Our positive results for the first-best gains-from-trade and first-best social welfare stand in stark contrast to those of~\cite{hajiaghayi2025bilateral}, who show pessimistic inapproximability results under a profit-maximizing broker with full distributional knowledge. 
Moreover, our results are practically appealing in the sense that the broker can forgo complex data collection and still achieve a provable share of several optimal economic outcomes, benefiting both herself and society. This could inform the design of online marketplaces and decentralized trading platforms, as platforms can still facilitate efficient trade between users without requiring invasive data practices.

In what follows, we first review related work in Section~\ref{sec:rel}, and describe our model in Section~\ref{sec:model}. We present our results in Section~\ref{sec:overview}, deferring some proofs to the appendix. Due to space constraint, some of the proofs in Appendix A are deferred to the full version.

\section{Related Work}\label{sec:rel}
As mentioned above, the bilateral trade problem was introduced in the seminal work of~\citet{myerson1983efficient}, who show that no incentive-compatible, individually-rational, and budget-balanced mechanism can achieve ex-post efficiency. As such, later work has focused on the \textit{approximability} of the first-best efficiency by such mechanisms. In fact, due to their simplicity, much of the literature has studied the approximability of the first-best efficiency via \textit{posted-pricing mechanisms} (\cite{mcafee2008gains},~\cite{blumrosen2016approximating},~\cite{blumrosen2016almost},~\cite{deng2022approximately},~\cite{kang2022fixed},~\cite{cai2023fixed},~\cite{liu2023improved}), including works studying such approximability in the context of general two-sided markets (\cite{brustle2017approximating},~\cite{colini2017fixed}).

Notably, all the literature above assumes that the trade is governed by a benevolent principal. However, in the same work where they introduce bilateral trade,~\citet{myerson1983efficient} introduce a variation of bilateral trade where a \textit{profit-maximizing broker} governs the trade. In spite of this natural extension, there had not been much work on this model until~\citet{hajiaghayi2025bilateral} initiated the study of the approximability of the first-best gains-from-trade in bilateral trade with a strategic broker.

Orthogonally, other works have considered markets under limited information. For example,~\citet{dhangwatnotai2010revenue} initiated the study of revenue maximization in one-sided markets when the mechanism only has access to a single sample from the buyer's distribution, with follow-up works by~\citet{huang2015making},~\citet{fu2015random}, and~\citet{goldner2016prior}. Analogously,~\citet{babaioff2020bulow},~\cite{dutting2021efficient},~\citet{kang2022fixed}, ~\citet{cai2023fixed}, and~\citet{liu2023improved} study the approximability of the first-best efficiency in two-sided markets when the mechanism only has access to a single sample from the distribution of each agent. These latter works have essentially resolved the optimal approximation ratios to the first-best gains-from-trade and first-best social welfare in the symmetric and stochastic dominance settings for single-sample bilateral trade without a broker, and this paper draws an analogy to this line of work in single-sample bilateral trade \emph{with} a broker.

\begin{table*}[t]
    \begin{minipage}{\textwidth}
    \renewcommand{\thempfootnote}{\fnsymbol{mpfootnote}}
    \begin{center}
        \begin{tabular}{|c|c|c|c||c|c|}
            \hline
            \multirow{2}{*} & \multicolumn{3}{c||}{With Broker\footnote{with respect to the mechanism that draws $p \sim F, q \sim G$ and posts $p$ to the buyer and $q$ to the seller}} & \multicolumn{2}{c|}{Without Broker}\\
            \cline{2-6}
             & GFT & SW & Profit & GFT\footnote{with respect to the mechanism that draws $p \sim F$ and posts $p$ to both buyer and seller} & SW\footnote{with respect to the mechanism that draws $q \sim G$ and posts $q$ to both buyer and seller}\\ 
            \hline
             Lower Bound & $0.1254$ & $0.1321$ & $1 / 180$\footnote{if $F$ and $G$ have monotone hazard rates} & $1 / 4$ (\cite{babaioff2020bulow}) & $1 / 2$ (\cite{cai2023fixed})\footnote{This lower bound holds for general problem instances, even without the stochastic-dominance assumption.}
             \\
             \hline
             Upper Bound & $7 / 48$ & $1 / 6$ & --- & $7 / 16$ (\cite{babaioff2020bulow}) & --- \\
             \hline
        \end{tabular}
    \end{center}
    \end{minipage}
    \caption{Our results for the stochastic dominance setting.~\citet{cai2023fixed} show an upper bound on the approximability of the first-best social welfare by their mechanism under general problem instances, but to the best of our knowledge, it remains open whether this upper bound holds in the special case of stochastic dominance.
    }
    \label{tab:asym-results}
\end{table*}

\section{Model}\label{sec:model}
\paragraph{Bilateral trade with a broker}
In bilateral trade, a seller wishes to sell a single item to a buyer. The buyer's valuation $v$ and the seller's valuation $c$ are drawn independently from distributions $F$ and $G,$ respectively, both of which are bounded, absolutely continuous distributions supported on nonnegative real values.\footnote{By Remark 2 in~\citet{deng2022approximately}, we can make this assumption without loss of generality.} The buyer and seller report valuations $\tilde{v}$ and $\tilde{c},$ respectively, to the broker, who then determines whether the item is transferred, and if so, how much each agent pays or receives.\footnote{Although a general mechanism may allow the buyer and seller to interact with the broker in other ways, we are only interested in incentive-compatible mechanisms, so by the revelation principle (\cite{myerson1981optimal}), we can consider \textit{direct mechanisms} without any loss of generality.}  Importantly, the broker only has access to a single sample from each agent's distribution on which to base her decisions. Formally, the broker runs a mechanism $\Mec = (x, p_b, p_s),$ where $x(\tilde{v}, \tilde{c}), p_b(\tilde{v}, \tilde{c}),$ and $p_s(\tilde{v}, \tilde{c})$ are the probability that the item is transferred, the expected payment of the buyer to the broker, and the expected payment of the broker to the seller, respectively, given the agents' reported valuations $\tilde{v}$ and $\tilde{c}.$

\paragraph{Desiderata} Within the class of all possible mechanisms, we are especially interested in those that satisfy several nice properties. First, we say that a mechanism $\Mec = (x, p_b, p_s)$ is \textit{Bayes-Nash incentive compatible} (BNIC) if it is in each agent's best interest to report his valuation truthfully, assuming truthful reporting by the other agent, \ie
\begin{align*}
    \Exu{c \sim G}{vx(v, c) - p_b(v, c)} &\ge \Exu{c \sim G}{vx(\tilde{v}, c) - p_b(\tilde{v}, c)} \text{ and}\\
    \Exu{v \sim F}{p_s(v, c) - cx(v, c)} &\ge \Exu{v \sim F}{p_s(v, \tilde{c}) - cx(v, \tilde{c})}
\end{align*}
for all $v, \tilde{v}, c,$ and $\tilde{c}.$ We further say that a mechanism is \textit{dominant-strategy incentive compatible} (DSIC) if it is in each agent's best interest to report his valuation truthfully, \textit{regardless of what the other agent does}, \ie
\begin{align*}
    vx(v, c) - p_b(v, c) &\ge vx(\tilde{v}, c) - p_b(\tilde{v}, c) \text{ and}\\
    p_s(v, c) - cx(v, c) &\ge p_s(v, \tilde{c}) - cx(v, \tilde{c})
\end{align*}
for all $v, \tilde{v}, c,$ and $\tilde{c}.$ Since we are only interested in BNIC mechanisms, we henceforth write $v$ and $c$ instead of $\tilde{v}$ and $\tilde{c},$ respectively, for the reported valuations of the agents.

To incentivize the agents to participate in a mechanism, their marginal utilities from participating should be nonnegative. To this end, we say that a mechanism is \textit{individually-rational} (IR) if 
\begin{align*}
    \Exu{v \sim F, c \sim G}{vx(v, c) - p_b(v, c)} &\ge 0 \text{ and}\\
    \Exu{v \sim F, c \sim G}{p_s(v, c) - cx(v, c)} &\ge 0.
\end{align*}
We further say that a mechanism is \textit{ex-post individually-rational} if the inequalities above hold pointwise, \ie for all $v$ and $c$:
\begin{align*}
    vx(v, c) - p_b(v, c) &\ge 0 \text{ and}
    \\
    p_s(v, c) - cx(v, c) &\ge 0
\end{align*}

In light of the desiderata above, we are especially interested in \textit{posted-pricing mechanisms} in which the broker offers take-it-or-leave-it prices $p$ and $q$ to the buyer and seller, respectively, and the item is exchanged if and only if $v \ge p$ and $c \le q.$ Note that posted-pricing mechanisms are DSIC and ex-post IR. 

\paragraph{Metrics and benchmarks}
We are interested in three metrics and their corresponding benchmarks for a mechanism $\Mec = (x, p_b, p_s).$ First, we define the \textit{gains-from-trade} (\GFT) of $\Mec$ as the expected marginal increase in the welfare of society from the trade, \ie
\begin{align*}
    \GFT &= \Exu{v \sim F, c \sim G}{x(v, c)(v - c)},
\end{align*}
and the \textit{social welfare} (\SW) of $\Mec$ as the expected welfare of society from the trade, \ie 
\begin{align*}
    \SW &= \Exu{c \sim G}{c} + \GFT.
\end{align*}
Our benchmarks for these two measures of efficiency will be the gains-from-trade (resp. social welfare) achieved by any \textit{ex-post efficient} mechanism, \ie a mechanism that enforces trade whenever the buyer's valuation $v$ weakly exceeds the seller's valuation $c.$ To this end, we define the first-best gains-from-trade as
\begin{align*}
    \FB &= \Exu{v \sim F, c \sim G}{(v - c)\mathbf{1}\{v \ge c\}}
\end{align*}
and the first-best social welfare as
\begin{align*}
    \FBW &= \Exu{c \sim G}{c} + \FB.
\end{align*}
Finally, we define the profit of $\Mec$ to be
\begin{align*}
    \Pro &= \Exu{v \sim F, c \sim G}{p_b(v, c) - p_s(v, c)}.
\end{align*}
We will be comparing the profit of our mechanisms to the optimal profit $\Pro^*$ achieved by any BNIC and IR mechanism with full distributional knowledge.\footnote{In fact, our lower bounds for the profit of our mechanisms hold even when we replace $\Pro^*$ with the stronger benchmark of $\FB.$ Indeed, although~\cite{myerson1983efficient} derive an integral expression for $\Pro^*$ when $F$ and $G$ are regular, it seems technically challenging to compare the profits of our mechanisms to this integral directly.} 

As one final piece of notation, we use $\GFT(p, q)$ to denote the gains-from-trade induced by offering price $p$ to the buyer and price $q$ to the seller, where we define $\GFT(p, q) = 0$ if $p < q.$ We use $\Pro(p, q)$ analogously.

\paragraph{Distribution types}
Lastly, we introduce some relevant definitions related to the agents' distributions $F$ and $G.$ We say that the buyer and seller are \textit{symmetric} if $F = G,$ and we say that the buyer's distribution $F$ \emph{stochastically dominates} the seller's distribution $G$ if $G(x) \ge F(x)$ for all $x \ge 0.$ The buyer's distribution satisfies the \textit{monotone hazard rate} (MHR) property if $(1 - F(x)) / F'(x)$ is nonincreasing, and the seller's distribution satisfies the MHR property if $G(x) / G'(x)$ is nondecreasing. Finally, in the symmetric setting, we say that $F$ has a \textit{doubly monotone hazard rate} if $F$ satisfies the MHR property both as the buyer's distribution and as the seller's distribution.

\section{Main Results}\label{sec:overview}
We now present our main results.
We consider the natural mechanism that draws samples $p \sim F, q \sim G$ and offers price $p$ to the buyer and price $q$ to the seller,\footnote{If the agents are symmetric, our mechanism offers price $\max(p, q)$ to the buyer and $\min(p, q)$ to the seller. Otherwise, we assume that no trade occurs whenever $p < q.$} and we are interested in the extent to which this mechanism approximates the first-best gains-from-trade, first-best social welfare, and optimal profit. 
Note that by randomizing over our mechanism and mechanisms for bilateral trade without a broker, we can interpolate between the lower bounds for gains-from-trade/social welfare and profit in the settings with and without a broker.

As a fundamental preliminary result, we first characterize the profit of our mechanism:
\begin{lemma}
\label{lem:sample-pft}
    The mechanism that draws samples $p \sim F, q \sim G$ and offers price $p$ to the buyer and price $q$ to the seller achieves
    \begin{align*}
        \Exu{p \sim F, q \sim G}{\Pro(p, q)} = \frac{1}{4}\int_0^\infty G(x)^2(1 - F(x))^2 \, dx.
    \end{align*}
\end{lemma}
\begin{proof}
    We have that
    \begin{align*}
        \Exu{p \sim F, q \sim G}{\Pro(p, q)} &= \int_0^{\infty} \int_q^{\infty} (p - q)(1 - F(p))G(q) \, dF(p) \, dG(q)
        \\
        &= \int_0^{\infty} G(q)\parans{\int_q^{\infty} (p - q)(1 - F(p)) \, dF(p)} dG(q)
        \\
        &= \int_0^{\infty} G(q)\parans{\int_q^{\infty} (p - q)(1 - F(p))F'(p) \, dp} dG(q).
    \end{align*}
    Integrating by parts over $(1-F(p))F'(p)$ and $(p-q)$, we have
    \begin{align*}
        \Exu{p \sim F, q \sim G}{\Pro(p, q)} &= \int_0^{\infty} G(q)\left(-\frac{(1 - F(p))^2}{2}(p - q) \Bigg|_q^\infty\right. + \left.\frac{1}{2}\int_q^\infty (1 - F(p))^2 \, dp\right) dG(q)
        \\
        &= \frac{1}{2}\int_0^{\infty} G(q)\parans{\int_q^\infty (1 - F(p))^2 \, dp} dG(q)
        \\
        &= \frac{1}{2}\int_0^{\infty} G(q)\parans{\int_q^\infty (1 - F(p))^2 \, dp} G'(q) \, dq.
    \end{align*}
    Integrating by parts again over $G(q)G'(q)$ and $\displaystyle \int_q^\infty (1-F(p))^2 \, dp$, we finally obtain that
    \begin{align*}
        \Exu{p \sim F, q \sim G}{\Pro(p, q)} &= \frac{1}{2}\left(\frac{G(q)^2}{2}\int_q^\infty (1 - F(p))^2 \, dp \Bigg|_0^\infty\right. + \left.\frac{1}{2}\int_0^\infty G(q)^2(1 - F(q))^2 \, dq\right)\\
        &= \frac{1}{4}\int_0^\infty G(x)^2(1 - F(x))^2 \, dx.
    \end{align*}
\end{proof}
Next, we characterize the gains-from-trade of our mechanism, generalizing the characterization of~\citet{kang2022fixed} for the gains-from-trade obtained by offering price $p \sim F$ to both buyer and seller under symmetric agents in bilateral trade without a broker:
\begin{lemma}
\label{lem:sample-gft}
    The mechanism that draws samples $p \sim F, q \sim G$ and offers price $p$ to the buyer and price $q$ to the seller achieves
    \begin{align*}
         \Exu{p \sim F, q \sim G}{\GFT(p, q)} &= \frac{1}{4}\int_0^\infty G(x)^2(1 - F(x))^2 \, dx + \frac{1}{2}\int_0^\infty G(q)^2 \int_q^{\infty}(1-F(x)) \, dx \, dF(q)\\
         &+ \frac{1}{2}\int_0^\infty (1 - F(p))^2 \int_0^p G(x) \, dx \, dG(p).
    \end{align*}
\end{lemma}
The proof follows from decomposing the gains-from-trade into the profit and some additional terms involving integrals over $1 - F(x)$ and $G(x),$ followed by some algebraic manipulation.

\subsection{Symmetric Setting}
With these preliminary lemmas in hand, we now turn our attention to the symmetric setting. To begin with, we use Lemma~\ref{lem:sample-gft} to characterize the gains-from-trade of our single-sample mechanism in the symmetric setting:
\begin{lemma}
\label{lem:sample-sym-gft}
    In the symmetric setting, the mechanism that draws samples $p, q \sim F$ and offers price $\max(p, q)$ to the buyer and price $\min(p, q)$ to the seller achieves
    \begin{align*}
         \GFT &= \frac{1}{3}\int_0^\infty F(x)(1 - F(x)) \, dx - \frac{1}{6}\int_0^\infty F(x)^2(1 - F(x))^2 \, dx.
    \end{align*}
\end{lemma}
\begin{proof}
We have that
\begin{align*}
    \GFT &= \Exu{p, q \sim F}{\GFT(p, q)\mathbf{1}\{p \ge q\}} + \Exu{p, q \sim F}{\GFT(q, p)\mathbf{1}\{q \ge p\}}\\
    &= 2\Exu{p, q \sim F}{\GFT(p, q)\mathbf{1}\{p \ge q\}}\\
    &= 2\Exu{p, q \sim F}{\GFT(p, q)}\\
    &= \frac{1}{2}\int_0^\infty F(x)^2(1 - F(x))^2 \, dx + \int_0^\infty F(q)^2 \int_q^{\infty}(1-F(x)) \, dx \, dF(q)\\
    &+ \int_0^\infty (1 - F(p))^2 \int_0^p F(x) \, dx \, dF(p),
\end{align*}
where the last equality holds by Lemma~\ref{lem:sample-gft}. Evaluating the first integral, we obtain that
\begin{align*}
    \int_0^\infty F(q)^2 \int_q^{\infty}(1-F(x)) \, dx \, dF(q) &= \int_0^\infty F(q)^2 \int_q^{\infty}(1-F(x)) \, dx \, F'(q) \, dq
    \\
    &= \frac{F(q)^3}{3}\int_q^{\infty}(1-F(x)) \, dx \Bigg|_0^\infty + \frac{1}{3}\int_0^\infty F(q)^3(1 - F(q)) \, dq\\
    &= \frac{1}{3}\int_0^\infty F(x)^3(1 - F(x)) \, dx, 
\end{align*}
where the second equality holds by integrating by parts over $F(q)^2F'(q)$ and $\displaystyle\int_q^\infty (1 - F(x)) \, dx.$ Evaluating the second integral, we have 
\begin{align*}
    \int_0^\infty (1 - F(p))^2 \int_0^p F(x) \, dx \, dF(p) &= \int_0^\infty (1 - F(p))^2 \int_0^p F(x) \, dx \, F'(p) \, dp
    \\
    &= -\frac{(1 - F(p))^3}{3}\int_0^p F(x) \, dx\Bigg|_0^\infty + \frac{1}{3}\int_0^\infty (1 - F(p))^3F(p) \, dp\\
    &= \frac{1}{3}\int_0^\infty (1 - F(x))^3F(x) \, dx,
\end{align*}
where the second equality holds by integrating by parts over $(1 - F(p))^2F'(p)$ and $\displaystyle\int_0^p F(x) \, dx.$ Thus, we have that
\begin{align*}
    \GFT &= \frac{1}{2}\int_0^\infty F(x)^2(1 - F(x))^2 \, dx + \frac{1}{3}\int_0^\infty F(x)^3(1 - F(x)) \, dx + \frac{1}{3}\int_0^\infty (1 - F(x))^3F(x) \, dx\\
    &= \frac{1}{3}\int_0^\infty F(x)(1 - F(x)) \, dx - \frac{1}{6}\int_0^\infty F(x)^2(1 - F(x))^2 \, dx.
\end{align*}
\end{proof}
We next prove the following lower bound on the approximability of the first-best gains-from-trade achieved by our mechanism:
\begin{theorem}
\label{thm:sample-const-apx-gft}
    In the symmetric setting, the mechanism that draws samples $p, q \sim F$ and offers price $\max(p, q)$ to the buyer and price $\min(p, q)$ to the seller achieves $\GFT \ge \nicefrac{7}{24}\cdot \FB.$
\end{theorem}
\begin{proof}
By Lemma~\ref{lem:sample-sym-gft}, we have that
\begin{align*}
    \GFT &= \frac{1}{3}\int_0^\infty F(x)(1 - F(x)) \, dx - \frac{1}{6}\int_0^\infty F(x)^2(1 - F(x))^2 \, dx\\
    &\ge \frac{1}{3}\int_0^\infty F(x)(1 - F(x)) \, dx - \frac{1}{6}\parans{\frac{1}{4}\int_0^\infty F(x)(1 - F(x)) \, dx}\\
    &= \frac{7}{24}\int_0^\infty F(x)(1 - F(x)) \, dx\\
    &= \frac{7}{24}\FB,
\end{align*}
where the inequality holds since $F(x)(1 - F(x)) \le 1 / 4$ for $x \ge 0.$
\end{proof}

We can similarly use Lemma~\ref{lem:sample-sym-gft} to prove the following lower bound on the approximability of the first-best social welfare achieved by our mechanism:
\begin{theorem}
\label{thm:sample-const-apx-sw}
    In the symmetric setting, the mechanism that draws samples $p, q \sim F$ and offers price $\max(p, q)$ to the buyer and price $\min(p, q)$ to the seller achieves $\SW \ge \nicefrac{2}{3} \cdot \FBW.$
\end{theorem}

On the other hand, our proof of a lower bound on approximability of the optimal profit achieved by our mechanism uses the following result from~\citet{hajiaghayi2025bilateral}:
\begin{lemma}{\cite{hajiaghayi2025bilateral}}\label{lem:hajia}
    Suppose $F$ is doubly MHR. Then, if $0 < \beta < \alpha < 1$ and $C \ge 1$ satisfy
    \begin{align*}
        \frac{C - 1}{2}\min\left(\frac{\alpha^2 - \beta^2}{2}, \alpha - \beta - \frac{\alpha^2 - \beta^2}{2}\right) \ge \max\left(\beta - \frac{\beta^2}{2}, \frac{1 - \alpha^2}{2}\right),
    \end{align*}
    we have that
    \begin{align*}
        \int_{F^{-1}(\beta)}^{F^{-1}(\alpha)} F(x)(1 - F(x)) \, dx \ge \frac{1}{C}\FB.
    \end{align*}
\end{lemma}
\begin{theorem}
\label{thm:sym-pro-const}
    Suppose $F$ is doubly MHR. Then, the mechanism that draws samples $p, q \sim F$ and offers price $\max(p, q)$ to the buyer and price $\min(p, q)$ to the seller achieves
    \begin{align*}
        \Pro \ge \frac{2}{55}\FB \ge \frac{2}{55}\Pro^*.
    \end{align*}
\end{theorem}
\begin{proof}
    Suppose we have constants $0 < \beta < \alpha < 1$ and $C \ge 1$ such that
    \begin{align*}
        \frac{C - 1}{2}\min\left(\frac{\alpha^2 - \beta^2}{2}, \alpha - \beta - \frac{\alpha^2 - \beta^2}{2}\right) \ge \max\left(\beta - \frac{\beta^2}{2}, \frac{1 - \alpha^2}{2}\right).
    \end{align*}
    Then, since $F$ is doubly-MHR, by Lemma~\ref{lem:hajia}, we have that
    \begin{align*}
        \int_{F^{-1}(\beta)}^{F^{-1}(\alpha)} F(x)(1 - F(x)) \, dx \ge \frac{1}{C}\FB,
    \end{align*}
    so
    \begin{align*}
        \Pro &= \Exu{p, q \sim F}{\Pro(p, q)\mathbf{1}\{p \ge q\}} +  \Exu{p, q \sim F}{\Pro(q, p)\mathbf{1}\{q \ge p\}}\\
        &= 2\Exu{p, q \sim F}{\Pro(p, q)\mathbf{1}\{p \ge q\}}\\
        &= 2\Exu{p, q \sim F}{\Pro(p, q)}\\
        &\ge \frac{1}{2}\int_{F^{-1}(\beta)}^{F^{-1}(\alpha)} F(x)^2(1 - F(x))^2 \, dx\\
        &\ge \frac{\min(\alpha(1 - \alpha), \beta(1 - \beta))}{2}\int_{F^{-1}(\beta)}^{F^{-1}(\alpha)} F(x)(1 - F(x)) \, dx\\
        &\ge \frac{\min(\alpha(1 - \alpha), \beta(1 - \beta))}{2C}\FB,
    \end{align*}
    where the first inequality holds by Lemma~\ref{lem:sample-pft}.
    By (roughly) numerically optimizing the parameters $\alpha, \beta,$ and $C,$ we take $\alpha = 4 / 5, \beta = 1 / 5,$ and $C = 11 / 5$ to obtain the desired result.
\end{proof}

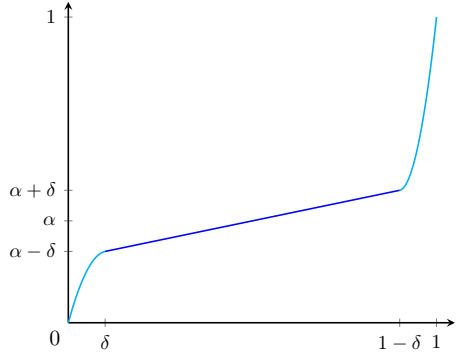
\begin{figure} 
\centering
\begin{tikzpicture}[scale=0.75]
\begin{axis}[
    axis lines = center,
    xmin=0, xmax=1.05,
    ymin=0, ymax=1.05,
    samples=200,
    domain=0:1,
    xtick={1/10, 9/10, 1},
    ytick={7/30, 1/3, 13/30, 1},
    xticklabels={$\delta$, $1 - \delta$, $1$},
    yticklabels={$\alpha - \delta$, $\alpha$, $\alpha + \delta$, $1$},
    ticklabel style={font=\small},
    thick,
    clip=false 
]

\addplot[domain=0:0.1, cyan] {7 / 3 * x - 125 / 6 * x * (x - 1 / 10)};
\addplot[domain=0.1:0.9, blue] {1 / 3 + 1 / 4 * (x - 1 / 2)};
\addplot[domain=0.9:1, cyan] {1 - 17 / 3 * (1 - x) + 325 / 6 * (1 - x) * (9 / 10 - x)};

\node[anchor=north east] at (axis cs:0,0) {$0$};

\end{axis}
\end{tikzpicture}
\caption{Graph of $H_{\alpha, \delta}(x)$}
\label{fig:dist}
\end{figure}

Now, taking a closer look at our proof of Theorem~\ref{thm:sample-const-apx-gft}, we observe that our lower bound can be made arbitrarily tight by taking $F(x) \approx 1 / 2.$ Following this intuition, we use polynomial interpolation to define a class of distributions $H_{\alpha, \delta}$ where $H_{\alpha, \delta} \approx \alpha$ (see Figure~\ref{fig:dist}):
\begin{definition}
    Let $0 < \alpha < 1$ and $\displaystyle 0 < \delta < \frac{\alpha}{2\alpha + 1}.$ Define $H_{\alpha, \delta, \QUAD}: [0, \delta] \rightarrow \mathbb{R}$ as 
    \begin{align*}
        H_{\alpha, \delta, \QUAD}(x) &= \frac{\alpha - \delta}{\delta}x + \frac{(2\alpha + 1)\delta - \alpha}{\delta^2(1 - 2\delta)}x(x - \delta)
    \end{align*}
    and 
    $H_{\alpha, \delta, \LIN}: [\delta, 1 - \delta] \rightarrow \mathbb{R}$ as 
    \begin{align*}
        H_{\alpha, \delta, \LIN}(x) &= \alpha + \frac{2\delta}{1 - 2\delta}\left(x - \frac{1}{2}\right).
    \end{align*}
\end{definition}
\begin{definition}
\label{def:dist}
     Let $0 < \alpha < 1$ and $\displaystyle 0 < \delta < \min\left(\frac{\alpha}{2\alpha + 1}, \frac{1 - \alpha}{3 - 2\alpha}\right).$ We define $H_{\alpha, \delta}: [0, 1] \rightarrow \mathbb{R}$ as
    \begin{align*}
        H_{\alpha, \delta}(x) &=
        \begin{cases}
            H_{\alpha, \delta, \QUAD}(x) & \text{ if $0 \le x \le \delta$}\\
            H_{\alpha, \delta, \LIN}(x) & \text{ if $\delta \le x \le 1 - \delta$}\\
            1 - H_{1 - \alpha, \delta, \QUAD}(1 - x) & \text{ if $1 - \delta \le x \le 1$}.
        \end{cases}
    \end{align*}
\end{definition}
We then use this class of distributions to obtain the following matching upper bound on the approximability of the first-best gains-from-trade achieved by our mechanism:
\begin{theorem}
\label{thm:sym-gft-ub}
    For any $\epsilon > 0$, there exists a symmetric instance of the bilateral trade problem such that the mechanism that draws samples $p, q \sim F$ and offers price $\max(p, q)$ to the buyer and price $\min(p, q)$ to the seller achieves $\GFT < \left(\nicefrac{7}{24} + \epsilon\right)\cdot \FB.$
\end{theorem}
\begin{proof}
Let $F = H_{1 / 2, \delta}$ for $0 < \delta < 1 / 4.$ Then, by Lemma~\ref{lem:sample-sym-gft}, we have that 
\begin{align}
    \GFT &= \frac{1}{3}\int_0^1 F(x)(1 - F(x)) \, dx - \frac{1}{6}\int_0^1 F(x)^2(1 - F(x))^2 \, dx \nonumber\\
    &=\frac{1}{3}\int_0^1 F(x)(1 - F(x)) \, dx - \frac{1}{6}\parans{\frac{1}{4}\int_0^1 F(x)(1 - F(x)) \, dx} \nonumber\\
    & +\frac{1}{6} \int_0^1 F(x)(1 - F(x))\left(\frac{1}{4} - F(x)(1 - F(x))\right) \, dx \nonumber\\
    &=\frac{7}{24}\int_0^1 F(x)(1 - F(x)) \, dx + \frac{1}{6} \int_0^1 F(x)(1 - F(x))\left(\frac{1}{4} - F(x)(1 - F(x))\right) \, dx\label{eq:12200810}.
\end{align}
We now analyze the second integral in~\eqref{eq:12200810}. Since $F(x) = 1 - F(1 - x)$ for all $1 / 2 \le x \le 1,$ by symmetry, we have that
\begin{align*}
    &\int_0^1 F(x)(1 - F(x))\left(\frac{1}{4} - F(x)(1 - F(x))\right) \, dx\\
    &= 2\int_0^{1 / 2} F(x)(1 - F(x))\left(\frac{1}{4} - F(x)(1 - F(x))\right) \, dx\\
    &\le \frac{1}{2}\int_0^{1 / 2} \left(\frac{1}{4} - F(x)(1 - F(x))\right) \, dx\\
    &= \frac{1}{2}\left(\int_0^\delta \left(\frac{1}{4} - F(x)(1 - F(x))\right) \, dx + \int_\delta^{1 / 2} \left(\frac{1}{4} - F(x)(1 - F(x))\right) \, dx\right)\\
    &\le \frac{1}{2}\left(\frac{1}{4}\delta + \int_\delta^{1 / 2} \left(\frac{1}{4} - F(x)(1 - F(x))\right) \, dx\right).
\end{align*}
Now, using the fact that $1 / 2 - \delta \le F(x) \le 1 / 2$ on $[\delta, 1 / 2],$ we upper-bound the integral above to obtain that
\begin{align*}
    \int_\delta^{1 / 2} \left(\frac{1}{4} - F(x)(1 - F(x))\right) \, dx &\le \int_\delta^{1 / 2} \left(\frac{1}{4} -\frac{1}{2}\left(\frac{1}{2} - \delta\right)\right) \, dx = \frac{1}{2}\delta\left(\frac{1}{2} - \delta\right).
\end{align*}
Thus,
\begin{align*}
    \int_0^1 F(x)(1 - F(x))\left(\frac{1}{4} - F(x)(1 - F(x))\right) \, dx &\le \frac{1}{2}\left(\frac{1}{4}\delta + \frac{1}{2}\delta\left(\frac{1}{2} - \delta\right)\right) \xrightarrow{\delta \rightarrow 0} 0, 
\end{align*}
so
\begin{align}
    \int_0^1 F(x)(1 - F(x))\left(\frac{1}{4} - F(x)(1 - F(x))\right) \, dx \xrightarrow{\delta \rightarrow 0} 0.\label{eq:12200808}
\end{align}
Finally, we have that
\begin{align}
    \FB &\ge \Exu{v, c \sim F}{(v - c)\mathbf{1}\left\{v \ge 1 - \delta, c \le \delta\right\}} \nonumber\\
    &\ge \Exu{v, c \sim F}{(1 - 2\delta)\mathbf{1}\left\{v \ge 1 - \delta, c \le \delta\right\}} \nonumber\\
    &\ge (1 - 2\delta)\Pru{v \sim F}{v \ge 1 - \delta}\Pru{c \sim G}{c \le \delta} \nonumber\\
    &= (1 - 2\delta)\left(\frac{1}{2} - \delta\right)^2 
    \ge \frac{1}{2} \cdot \left(\frac{1}{4}\right)^2\label{eq:12200809}
\end{align}
where the last inequality holds since $\delta < 1 / 4.$ In particular, $\FB$ is bounded away from $0.$  
Hence, by~\eqref{eq:12200810},~\eqref{eq:12200808}, and~\eqref{eq:12200809}, we have that $\displaystyle \frac{\GFT}{\FB} \xrightarrow{\delta \to 0} \frac{7}{24}$.
\end{proof}
A similar observation holds for our proof of Theorem~\ref{thm:sample-const-apx-sw}, except that we want to take $F(x) \approx 1$ instead: 
\begin{theorem}
\label{thm:sym-sw-ub}
    For any $\epsilon >0$, there exists a symmetric instance of the bilateral trade problem such that the mechanism that draws samples $p, q \sim F$ and offers price $\max(p, q)$ to the buyer and price $\min(p, q)$ to the seller achieves $\SW < \left(\nicefrac{2}{3} + \epsilon\right)\cdot \FBW.$
\end{theorem}
\begin{proof}
Let $F(x) = x^\delta$ for $\delta > 0.$ Then, 
\begin{align*}
    \FBW &= \Exu{c \sim F}{c} + \FB\\
    &= \int_0^1 (1 - F(x)) \, dx + \int_0^1 F(x)(1 - F(x)) \, dx\\
    &= \int_0^1 (1 - F(x)^2) \, dx\\
    &= \int_0^1 (1 - x^{2\delta}) \, dx\\
    &= 1 - \frac{1}{2\delta + 1} = \frac{2\delta}{2\delta + 1}.
\end{align*}
In addition, as in the proof of Theorem~\ref{thm:sample-const-apx-sw}, we have that
\begin{align*}
    \SW &= \frac{1}{6}\left(\int_0^1 (1 - F(x)^2)(F(x)^2 - 2F(x) + 4) \, dx + 2\int_0^1 (1 - F(x)) \, dx\right)\\
    &= \frac{1}{6}\left(\int_0^1 (4 - 2x^\delta - 3x^{2\delta} + 2x^{3\delta} - x^{4\delta}) \, dx + 2\int_0^1 (1 - x^\delta) \, dx\right)\\
    &= \frac{1}{6}\left(\left(4 - \frac{2}{\delta + 1} - \frac{3}{2\delta + 1} + \frac{2}{3\delta + 1} - \frac{1}{4\delta + 1}\right) + 2\left(1 - \frac{1}{\delta + 1}\right)\right)\\
    &= 1 - \frac{2 / 3}{\delta + 1} - \frac{1 / 2}{2\delta + 1} + \frac{1 / 3}{3\delta + 1} - \frac{1 / 6}{4\delta + 1}.
\end{align*}
Now, letting $h(\delta) = (\delta + 1)(2\delta + 1)(3\delta + 1)(4\delta + 1),$ we obtain that
\begin{align*}
    \SW &= \frac{1 + 10\delta + O(\delta^2)}{h(\delta)} - \frac{2(1 + 9\delta + O(\delta^2))}{3h(\delta)} - \frac{1 + 8\delta + O(\delta^2)}{2h(\delta)}\\
    &+ \frac{1 + 7\delta + O(\delta^2)}{3h(\delta)} - \frac{1 + 6\delta + O(\delta^2)}{6h(\delta)} = \frac{4\delta + O(\delta^2)}{3h(\delta)},
\end{align*}
where the big-O notation hides dominated terms with respect to $\delta$ as $\delta \rightarrow 0.$ Thus,
\begin{align*}
    \frac{\SW}{\FBW} &= \frac{2 + O(\delta)}{3(\delta + 1)(3\delta + 1)(4\delta + 1)} \xrightarrow{\delta \rightarrow 0} \frac{2}{3}.
\end{align*}
\end{proof}

\paragraph{Stochastic dominance setting}
In addition to the symmetric setting, we study the setting where the buyer's distribution stochastically dominates the seller's distribution. We first use Lemma~\ref{lem:sample-gft} to lower-bound the gains-from-trade of our single-sample mechanism in the stochastic dominance setting:
\begin{lemma}
\label{lem:sample-stoc-dom-gft}
    Suppose $F$ stochastically dominates $G.$ Then, the mechanism that draws samples $p \sim F, q \sim G$ and offers price $p$ to the buyer and price $q$ to the seller achieves
    \begin{align*}
         \GFT &\ge \frac{1}{12}\int_0^\infty (3G(x)^2(1 - F(x))^2 + 2F(x)^3(1 - F(x)) + 2(1 - G(x))^3G(x)) \, dx.
    \end{align*}
\end{lemma}
\begin{proof}
By Lemma~\ref{lem:sample-gft}, we have that
\begin{align*}
    \GFT &= \Exu{p \sim F, q \sim G}{\GFT(p, q)}\\
    &= \frac{1}{4}\int_0^\infty G(x)^2(1 - F(x))^2 \, dx + \frac{1}{2}\int_0^\infty G(q)^2 \int_q^{\infty}(1-F(x)) \, dx \, dF(q)\\
    &+ \frac{1}{2}\int_0^\infty (1 - F(p))^2 \int_0^p G(x) \, dx \, dG(p).
\end{align*}
Now, since $F$ stochastically dominates $G,$ we have that
\begin{align*}
    \int_0^\infty G(q)^2 \int_q^{\infty}(1-F(x)) \, dx \, dF(q) &\ge \int_0^\infty F(q)^2 \int_q^{\infty}(1-F(x)) \, dx \, dF(q)\\
    &= \frac{1}{3}\int_0^\infty F(x)^3(1 - F(x)) \, dx,
\end{align*}
where the last equality holds as in the proof of Lemma~\ref{lem:sample-sym-gft}.
Similarly, we have that
\begin{align*}
    \int_0^\infty (1 - F(p))^2 \int_0^p G(x) \, dx \, dG(p) &\ge \int_0^\infty (1 - G(p))^2 \int_0^p G(x) \, dx \, dG(p)\\
    &= \frac{1}{3}\int_0^\infty (1 - G(x))^3G(x) \, dx.
\end{align*}
Thus,
\begin{align*}
    \GFT &\ge \frac{1}{12}\int_0^\infty (3G(x)^2(1 - F(x))^2 + 2F(x)^3(1 - F(x)) + 2(1 - G(x))^3G(x)) \, dx.
\end{align*}
\end{proof}
Analogous to the symmetric setting, we next prove the following lower bound on the approximability of the first-best gains-from-trade achieved by our mechanism:
\begin{theorem}
\label{thm:sample-asym-const-gft}
    Suppose $F$ stochastically dominates $G.$ Then, the mechanism that draws samples $p \sim F, q \sim G$ and offers price $p$ to the buyer and price $q$ to the seller achieves $\GFT \ge C \cdot \FB,$ where $C \approx 0.1254.$
    
\end{theorem}
\begin{proof}
Consider the following optimization problem:
\begin{equation*}
    \begin{aligned}
        & \underset{f, g}{\text{minimize}}
        & & \frac{3g^2(1 - f)^2 + 2f^3(1 - f) + 2(1 - g)^3g}{12g(1 - f)} \\
        & \text{subject to}
        & & 0 \le f \le g \le 1.
    \end{aligned}
\end{equation*}
Solving this numerically,\footnote{\label{ft:wolfram}via WolframAlpha} we obtain that the optimum $C \approx 0.1254$ occurs at $(f, g) \approx (0.3909, 0.6091).$ Thus,
\begin{align*}
    \frac{1}{12}(3G(x)^2(1 - F(x))^2 + 2F(x)^3(1 - F(x)) + 2(1 - G(x))^3G(x)) \ge C \cdot G(x)(1 - F(x))
\end{align*}
for all $x \ge 0,$ so by Lemma~\ref{lem:sample-stoc-dom-gft},
\begin{align*}
    \GFT \ge C\int_0^\infty G(x)(1 - F(x)) \, dx = C \cdot \FB.
\end{align*}
\end{proof}

We similarly prove the following lower bound on the approximability of the first-best social welfare achieved by our mechanism:
\begin{theorem}
\label{thm:sample-asym-const-sw}
    Suppose $F$ stochastically dominates $G.$ Then, the mechanism that draws samples $p \sim F, q \sim G$ and offers price $p$ to the buyer and price $q$ to the seller achieves $\SW \ge \parans{\nicefrac{3 - \sqrt{2}}{12}} \cdot \FBW.$
\end{theorem}

Finally, we prove the following lower bound on the approximability of the optimal profit achieved by our mechanism, adapting a technique from~\citet{hajiaghayi2025bilateral} to the stochastic dominance setting and exploiting the MHR property to obtain our result.

\begin{theorem}
    Suppose $F$ stochastically dominates $G,$ where $F$ and $G$ both have monotone hazard rates. Then, the mechanism that draws samples $p \sim F, q \sim G$ and offers price $p$ to the buyer and price $q$ to the seller achieves $\Pro \ge \nicefrac{1}{180}\cdot \FB \ge \nicefrac{1}{180}\cdot \Pro^*.$
\end{theorem}
\begin{proof}
Suppose there exist constants $0 < \beta < \alpha < 1$ and $C \ge 1$ such that
\begin{align*}
    \int_{G^{-1}(\beta)}^{F^{-1}(\alpha)} G(x)(1 - F(x)) \, dx \ge \frac{1}{C}\FB.
\end{align*}
Then, by Lemma~\ref{lem:sample-pft},
\begin{align*}
    \Pro &= \Exu{p \sim F, q \sim G}{\Pro(p, q)}\\
    &\ge \frac{1}{4}\int_{G^{-1}(\beta)}^{F^{-1}(\alpha)} G(x)^2(1 - F(x))^2 \, dx\\ &\ge \frac{\beta(1 - \alpha)}{4}\int_{G^{-1}(\beta)}^{F^{-1}(\alpha)} G(x)(1 - F(x)) \, dx\\
    &\ge \frac{\beta(1 - \alpha)}{4C}\FB.
\end{align*}
(Note that $F^{-1}(\alpha) \ge F^{-1}(\beta) \ge G^{-1}(\beta),$ since $F$ stochastically dominates $G.$) 

Now, it suffices to find constants $\alpha, \beta,$ and $C$ such that
\begin{align}
    (C - 1)\int_{G^{-1}(\beta)}^{F^{-1}(\alpha)} G(x)(1-F(x)) \, dx &\ge \int_0^{G^{-1}(\beta)} G(x)(1-F(x)) \, dx + \int_{F^{-1}(\alpha)}^\infty G(x)(1-F(x)) \, dx,\label{eq:01190835}
\end{align}
since
\begin{align*}
    \FB 
    &= \int_{0}^\infty G(x)(1 - F(x)) \, dx
    \\ 
    &=
    \int_{0}^{G^{-1}(\beta)} G(x)(1 - F(x)) \, dx + \int_{G^{-1}(\beta)}^{F^{-1}(\alpha)} G(x)(1 - F(x)) \, dx + \int_{F^{-1}(\alpha)}^\infty G(x)(1 - F(x)) \, dx.
\end{align*}
For the integral on the left-hand side of~\eqref{eq:01190835}, we have that
\begin{align*}
    \int_{G^{-1}(\beta)}^{F^{-1}(\alpha)} G(x)(1 - F(x)) \, dx 
    &=
    \int_{G^{-1}(\beta)}^{F^{-1}(\alpha)}  G(x)\frac{1-F(x)}{F'(x)} F'(x) \, dx\\
    &\ge \frac{1- \alpha}{F'(F^{-1}(\alpha))}\int_{G^{-1}(\beta)}^{F^{-1}(\alpha)} G(x) \, dF(x)\\
    &\ge \frac{1- \alpha}{F'(F^{-1}(\alpha))}\int_{F^{-1}(\beta)}^{F^{-1}(\alpha)} F(x) \, dF(x)\\
    &= \frac{1- \alpha}{F'(F^{-1}(\alpha))}\left(\frac{\alpha^2 - \beta^2}{2}\right),
\end{align*}
where the first inequality holds since $F$ is MHR and the second inequality holds since $F$ stochastically dominates $G.$ Similarly, we have that
\begin{align*}
    \int_{G^{-1}(\beta)}^{F^{-1}(\alpha)}  G(x)(1 - F(x)) \, dx &\ge
    \int_{G^{-1}(\beta)}^{F^{-1}(\alpha)} \frac{G(x)}{G'(x)}(1-F(x))G'(x) \, dx\\
    &\ge \frac{\beta}{G'(G^{-1}(\beta))}\int_{G^{-1}(\beta)}^{F^{-1}(\alpha)} (1-F(x)) \, dG(x)\\
    &\ge \frac{\beta}{G'(G^{-1}(\beta))}\int_{G^{-1}(\beta)}^{G^{-1}(\alpha)} (1-G(x)) \, dG(x)\\
    &= \frac{\beta}{G'(G^{-1}(\beta))}\left(\alpha - \beta - \frac{\alpha^2 - \beta^2}{2}\right).
\end{align*}
Thus, we can lower-bound the left-hand side of~\eqref{eq:01190835} as
\begin{align}
    &(C - 1)\int_{G^{-1}(\beta)}^{F^{-1}(\alpha)} G(x)(1 - F(x)) \, dx \nonumber\\
    &\ge
    \frac{C - 1}{2}\parans{\frac{1-\alpha}{F'(F^{-1}(\alpha))}\left(\frac{\alpha^2 - \beta^2}{2}\right)
    +
    \frac{\beta}{G'(G^{-1}(\beta))}\left(\alpha-\beta - \frac{\alpha^2-\beta^2}{2}\right)
    } \nonumber\\
    &= \frac{C - 1}{2}\min\left(\frac{\alpha^2 - \beta^2}{2}, \alpha-\beta - \frac{\alpha^2-\beta^2}{2}\right)\cdot\left(\frac{1-\alpha}{F'(F^{-1}(\alpha))} + \frac{\beta}{G'(G^{-1}(\beta))}\right)
    \label{eq:01190933}
\end{align}
Now, consider the right-hand side of~\eqref{eq:01190835}. Upper-bounding the first integral, we obtain that
\begin{align*}
    \int_0^{G^{-1}(\beta)} G(x)(1 - F(x)) &\, dx  
    =
    \int_0^{G^{-1}(\beta)} \frac{G(x)}{G'(x)}(1-F(x)) G'(x) \, dx \nonumber
    \\
    &\le \frac{\beta}{G'(G^{-1}(\beta))}\int_0^{G^{-1}(\beta)} (1-F(x)) \, dG(x) \nonumber\\
    &\le \frac{\beta}{G'(G^{-1}(\beta))} \cdot \beta,
\end{align*}
and for the second integral, we similarly obtain that
\begin{align*}
    \int_{F^{-1}(\alpha)}^\infty G(x)(1 - F(x)) \, dx
    &= 
    \int_{F^{-1}(\alpha)}^\infty G(x)\frac{1-F(x)}{F'(x)} F'(x) \, dx\nonumber
    \\
    &\le
    \frac{1-\alpha}{F'(F^{-1}(\alpha))}\int_{F^{-1}(\alpha)}^\infty G(x) \, dF(x)\nonumber\\
    &\le \frac{1-\alpha}{F'(F^{-1}(\alpha))} \cdot (1 - \alpha).
\end{align*}
Hence, we can upper-bound the right-hand side in~\eqref{eq:01190835} as
\begin{align}
    &\int_0^{G^{-1}(\beta)} G(x)(1 - F(x)) \, dx + \int_{F^{-1}(\alpha)}^\infty G(x)(1 - F(x)) \, dx \nonumber\\
    &\quad\le \max\left(\beta, 1 - \alpha\right)\left(\frac{1 - \alpha}{F'(F^{-1}(\alpha))} + \frac{\beta}{G'(G^{-1}(\beta))}\right)\label{eq:01190953}.
\end{align}
Therefore, from~\eqref{eq:01190835}, we see that it suffices by~\eqref{eq:01190933} and~\eqref{eq:01190953} to find constants $\alpha, \beta, C$ such that
\begin{align*}
    \frac{C - 1}{2}\min\left(\frac{\alpha^2 - \beta^2}{2}, \alpha - \beta - \frac{\alpha^2 - \beta^2}{2}\right) \ge \max\left(\beta, 1 - \alpha\right).
\end{align*}
By (roughly) numerically optimizing the parameters $\alpha, \beta$ and $C,$ we take $\alpha = 2 / 3, \beta = 1 / 3,$ and $C = 5$ to obtain that
\begin{align*}
    \Pro \ge \frac{\beta(1 - \alpha)}{4C}\FB = \frac{1}{180}\FB \ge \frac{1}{180}\Pro^*,
\end{align*}
as desired.
\end{proof}


As in the symmetric setting, our proofs of Theorems~\ref{thm:sample-asym-const-gft} and~\ref{thm:sample-asym-const-sw} suggest that we can prove matching upper bounds  by taking $F(x) \approx \alpha_F$ and $G(x) \approx \alpha_G$ for suitably chosen constants $\alpha_F$ and $\alpha_G$ in each case. However, since the lower bound of Lemma~\ref{lem:sample-stoc-dom-gft} may not be tight, the optimal constants will not necessarily correspond to the solutions of the optimization problems in the proofs of Theorems~\ref{thm:sample-asym-const-gft} and~\ref{thm:sample-asym-const-sw}. Nevertheless, using the class of distributions from Definition~\ref{def:dist}, we show the following upper bounds on the approximability of the first-best gains-from-trade and first-best social welfare achieved by our mechanism:
\begin{theorem}
\label{thm:asym-ub}
    For any $\epsilon > 0$, there exist $F$ and $G,$ where $F$ stochastically dominates $G,$ such that the mechanism that draws samples $p \sim F, q \sim G$ and offers price $p$ to the buyer and price $q$ to the seller achieves 
    \begin{align*}
        \GFT < \left(\frac{7}{48} + \epsilon\right)\FB.
    \end{align*}
    
    Further, there exists $F$ and $G,$ where $F$ stochastically dominates $G,$ such that this mechanism achieves 
    \begin{align*}
        \SW < \left(\frac{1}{6} + \epsilon\right)\FBW.
    \end{align*}
\end{theorem}

\section{Concluding Remarks}
We leave several interesting directions for future work. First, it remains open whether our mechanism achieves a constant approximation to the optimal profit for general problem instances.
In addition, since increasing the number of samples does not necessarily improve the approximation ratio to the first-best gains-from-trade and first-best social welfare (as we give more power to the strategic broker), it would be interesting to characterize the exact trade-off between the gains-from-trade/social welfare and the number of samples available to the broker.

\section*{Acknowledgments}
This work is partially supported by DARPA QuICC, ONR MURI 2024 award on Algorithms, Learning, and Game Theory, Army-Research Laboratory (ARL) grant W911NF2410052, NSF AF:Small grants 2218678, 2114269, 2347322, and Royal Society grant IES\textbackslash R2\textbackslash 222170.

\bibliographystyle{ACM-Reference-Format}
\bibliography{ref}

\appendix

\section{Preliminaries for Upper-Bound Proofs}
When proving upper bounds for our mechanisms, we will often want problem instances consisting of distributions $H$ such that $H(x) \approx \alpha$ for some constant $0 < \alpha < 1.$ To that end, we repeat our definition of $H_{\alpha, \delta}$ here:
\begin{definition}
     Let $0 < \alpha < 1$ and $\displaystyle 0 < \delta < \min\left(\frac{\alpha}{2\alpha + 1}, \frac{1 - \alpha}{3 - 2\alpha}\right).$ We define $H_{\alpha, \delta}: [0, 1] \rightarrow \mathbb{R}$ as
    \begin{align*}
        H_{\alpha, \delta}(x) &=
        \begin{cases}
            H_{\alpha, \delta, \QUAD}(x) & \text{ if $0 \le x \le \delta$}\\
            H_{\alpha, \delta, \LIN}(x) & \text{ if $\delta \le x \le 1 - \delta$}\\
            1 - H_{1 - \alpha, \delta, \QUAD}(1 - x) & \text{ if $1 - \delta \le x \le 1$}.
        \end{cases}
    \end{align*}
\end{definition}
\label{apx:dist}
We first show that $H_{\alpha, \delta}$ is well-defined and differentiable:
\begin{lemma}
\label{lem:welldef}
    Let $0 < \alpha < 1$ and $\displaystyle 0 < \delta < \min\left(\frac{\alpha}{2\alpha + 1}, \frac{1 - \alpha}{3 - 2\alpha}\right).$ Then, $H_{\alpha, \delta}$ is well-defined and differentiable. 
\end{lemma}
\begin{proof}
    Note that $H_{1 - \alpha, \delta, \QUAD}(1 - x)$ is well-defined on $[1 - \delta, 1],$ since $\displaystyle\delta < \frac{1 - \alpha}{3 - 2\alpha} = \frac{1 - \alpha}{2(1 - \alpha) + 1}.$ In addition, note the following:
    \begin{enumerate}
        \item $H_{\alpha, \delta, \QUAD}(x)$ is the unique quadratic satisfying $H_{\alpha, \delta, \QUAD}(0) = 0, H_{\alpha, \delta, \QUAD}(\delta) = \alpha - \delta,$ and $H'_{\alpha, \delta, \QUAD}(\delta) = 2\delta / (1 - 2\delta).$
        \item $H_{\alpha, \delta, \LIN}(x)$ is the unique linear function satisfying $H_{\alpha, \delta, \LIN}(\delta) = \alpha - \delta, H_{\alpha, \delta, \LIN}(1 - \delta) = \alpha + \delta,$ and $H'_{\alpha, \delta, \LIN}(\delta) = 2\delta / (1 - 2\delta) = H'_{\alpha, \delta, \LIN}(1 - \delta).$
        \item $Q(x) = 1 - H_{1 - \alpha, \delta, \QUAD}(1 - x)$ is the unique quadratic satisfying $Q(1) = 1, Q(1 - \delta) = \alpha + \delta,$ and $Q'(1 - \delta) = 2\delta / (1 - 2\delta).$
    \end{enumerate}
    Thus, since the different components and their derivatives match at the boundary points $\delta$ and $1-\delta,$ we have that $H_{\alpha, \delta}$ is well-defined and differentiable.
\end{proof}

Now, we want to show that $H_{\alpha, \delta}$ is a valid CDF. However, it is unclear whether $H_{\alpha, \delta, \QUAD}$ is strictly increasing. The following lemma verifies this:
\begin{lemma}
\label{lem:quad-inc}
    Let $0 < \alpha < 1$ and $\displaystyle 0 < \delta < \frac{\alpha}{2\alpha + 1}.$ Then, $H_{\alpha, \delta, \QUAD}$ is strictly increasing.
\end{lemma}
\begin{proof}
    Taking the derivative of $H_{\alpha, \delta, \QUAD},$ we obtain that
    \begin{align*}
        H_{\alpha, \delta, \QUAD}'(x) &= \frac{\alpha - \delta}{\delta} + \frac{(2\alpha + 1)\delta - \alpha}{\delta^2(1 - 2\delta)}(2x - \delta).
    \end{align*}
    Now, $H_{\alpha, \delta, \QUAD}'(x) > 0$ is equivalent to
    \begin{align*}
        2x - \delta < \frac{\delta(1 - 2\delta)(\delta - \alpha)}{(2\alpha + 1)\delta - \alpha}.
    \end{align*}
    (Note that $(2\alpha + 1)\delta - \alpha < 0$ and $1 - 2\delta > 0,$ since $\delta < \alpha / (2\alpha + 1) < 1 / 2.$). Rewriting, we obtain that
    \begin{align*}
        x < \delta\left(1 - \frac{\delta^2}{(2\alpha + 1)\delta - \alpha}\right),
    \end{align*}
    so it suffices if $x \le \delta$, which is true on the domain of $H_{\alpha, \delta, \QUAD}.$ Thus, $H_{\alpha, \delta, \QUAD}$ is strictly increasing.
\end{proof}

Finally, we combine our results above to show that $H_{\alpha, \delta}$ is a valid CDF:
\begin{lemma}
    Let $0 < \alpha < 1$ and $\displaystyle 0 < \delta < \min\left(\frac{\alpha}{2\alpha + 1}, \frac{1 - \alpha}{3 - 2\alpha}\right).$ Then, $H_{\alpha, \delta}$ is a CDF that is differentiable everywhere on its domain.
\end{lemma}
\begin{proof}
     By Lemma~\ref{lem:welldef}, $H_{\alpha, \delta}$ is well-defined and differentiable. In addition, $H_{\alpha, \delta}(0) = 0$ and $H_{\alpha, \delta}(1) = 1.$ Thus, to show that $H_{\alpha, \delta}$ is a valid CDF, it suffices to show that $H_{\alpha, \delta}$ is strictly increasing. Clearly, $H_{\alpha, \delta, \LIN}$ is strictly increasing, and by Lemma~\ref{lem:quad-inc}, $H_{\alpha, \delta, \QUAD}$ is strictly increasing on $[0, \delta]$ and $1 - H_{1 - \alpha, \delta, \QUAD}$ is strictly increasing on $[1 - \delta, 1],$ which completes the proof.
\end{proof}

Now, since we will use the class of distributions above to construct problem instances that show upper bounds for our mechanism in the stochastic dominance setting, the following lemma will be helpful:
\begin{lemma}
\label{lem:stoc-dom}
    Let $0 < \alpha_F \le \alpha_G < 1,$ and let $$0 < \delta <\min\left(\frac{\alpha_F}{2\alpha_F + 1}, \frac{1 - \alpha_F}{3 - 2\alpha_F}, \frac{\alpha_G}{2\alpha_G + 1}, \frac{1 - \alpha_G}{3 - 2\alpha_G}\right).$$ Then, $F = H_{\alpha_F, \delta}$ stochastically dominates $G = H_{\alpha_G, \delta}.$
\end{lemma}
\begin{proof}
    We have that
    \begin{align*}
        H_{\alpha_G, \delta}(x) &= \frac{\alpha_G - \delta}{\delta}x + \frac{(2\delta - 1)\alpha_G + \delta}{\delta^2(1 - 2\delta)}x(x - \delta) \nonumber\\
        &\ge \frac{\alpha_F - \delta}{\delta}x + \frac{(2\delta - 1)\alpha_F + \delta}{\delta^2(1 - 2\delta)}x(x - \delta) \nonumber\\
        &= H_{\alpha_F, \delta}(x)
    \end{align*}
    for all $0 \le x \le \delta,$ where the inequality holds since $2\delta - 1 < 0$ and $x - \delta \le 0.$ Thus, $G(x) \ge F(x)$ on $[0, \delta],$ and letting $1 - \alpha_G$ and $1 - \alpha_F$ take the roles of $\alpha_F$ and $\alpha_G,$ respectively, in the above inequality, we obtain that 
    \begin{align*}
        G(x) = 1 - H_{1 - \alpha_G, \delta}(1 - x) \ge 1 - H_{1 - \alpha_F, \delta}(1 - x) = F(x)
    \end{align*}
    on $[1 - \delta, 1].$ Finally, we have that 
    \begin{align*}
        G(x) = F(x) + (\alpha_G - \alpha_F) \ge F(x)
    \end{align*}
    on $[\delta, 1 - \delta].$
\end{proof}


Next, we seek to upper-bound the gains-from-trade of our mechanism when $F = H_{\alpha_F, \delta}$ and $G = H_{\alpha_G, \delta},$ as $\delta \rightarrow 0.$ The following lemma will be helpful in this regard:
\begin{lemma}
\label{lem:err-bnd}
    Let $0 < \alpha_F \le \alpha_G < 1$ and $\displaystyle 0 < \delta < \min\left(\frac{\alpha_F}{2\alpha_F + 1}, \frac{\alpha_G}{2\alpha_G + 1}\right).$ Then,
    \begin{align*}
        \int_0^\delta (H_{\alpha_G, \delta, \QUAD}(x)^2 - H_{\alpha_F, \delta, \QUAD}(x)^2)H_{\alpha_F, \delta, \QUAD}'(x) \, dx &= \frac{(\alpha_G - \alpha_F)(\alpha_G + \alpha_F)\alpha_F}{3} + O(\delta),
    \end{align*}
    where the big-O notation hides dominated terms with respect to $\delta$ as $\delta \rightarrow 0.$
\end{lemma}
\begin{proof}
    For $0 < \alpha < 1,$ we have that
    \begin{align*}
        H_{\alpha, \delta, \QUAD}(x) &= \frac{\alpha - \delta}{\delta}x + \frac{(2\alpha + 1)\delta - \alpha}{\delta^2(1 - \delta)}x(x - \delta)\\
        &= O\left(\frac{\alpha}{\delta}\right)x + \left(O\left(\frac{\alpha}{\delta}\right)x - O\left(\frac{\alpha}{\delta^2}\right)x^2\right)\\
        &= O\left(\frac{2\alpha}{\delta}\right)x - O\left(\frac{\alpha}{\delta^2}\right)x^2
    \end{align*}
    and that
    \begin{align*}
        H_{\alpha, \delta, \QUAD}'(x) &= \frac{\alpha - \delta}{\delta} + \frac{(2\alpha + 1)\delta - \alpha}{\delta^2(1 - \delta)}(2x - \delta)\\
        &= O\left(\frac{\alpha}{\delta}\right) + \left(O\left(\frac{\alpha}{\delta}\right) - O\left(\frac{2\alpha}{\delta^2}\right)x\right)\\
        &= O\left(\frac{2\alpha}{\delta}\right) - O\left(\frac{2\alpha}{\delta^2}\right)x.
    \end{align*}
    Thus,
    \begin{align*}
        &\int_0^\delta (H_{\alpha_G, \delta, \QUAD}(x)^2 - H_{\alpha_F, \delta, \QUAD}(x)^2)H_{\alpha_F, \delta, \QUAD}'(x) \, dx\\
        &\quad= \int_0^\delta (H_{\alpha_G, \delta, \QUAD}(x) - H_{\alpha_F, \delta, \QUAD}(x))(H_{\alpha_G, \delta, \QUAD}(x) + H_{\alpha_F, \delta, \QUAD}(x))H_{\alpha_F, \delta, \QUAD}'(x) \, dx\\
        &\quad= \int_0^\delta \left(\left(O\left(\frac{2(\alpha_G - \alpha_F)}{\delta}\right)x - O\left(\frac{\alpha_G - \alpha_F}{\delta^2}\right)x^2\right)\left(O\left(\frac{2(\alpha_G + \alpha_F)}{\delta}\right)x - O\left(\frac{\alpha_G + \alpha_F}{\delta^2}\right)x^2\right)\right.\\
        &\quad\quad\left.\left(O\left(\frac{2\alpha_F}{\delta}\right) - O\left(\frac{2\alpha_F}{\delta^2}\right)x\right) \, dx\right).
    \end{align*}
    Finally, expanding the integrand and collecting like terms, we obtain that
    \begin{align*}
        &\int_0^\delta (H_{\alpha_G, \delta, \QUAD}(x)^2 - H_{\alpha_F, \delta, \QUAD}(x)^2)H_{\alpha_F, \delta, \QUAD}'(x) \, dx\\
        &\quad= \int_0^\delta \left(O\left(\frac{8(\alpha_G - \alpha_F)(\alpha_G + \alpha_F)\alpha_F}{\delta^3}\right)x^2 
        - O\left(\frac{16(\alpha_G - \alpha_F)(\alpha_G + \alpha_F)\alpha_F}{\delta^4}\right)x^3\right.\\
        &\quad\quad+\left. O\left(\frac{10(\alpha_G - \alpha_F)(\alpha_G + \alpha_F)\alpha_F}{\delta^5}\right)x^4 - O\left(\frac{2(\alpha_G - \alpha_F)(\alpha_G + \alpha_F)\alpha_F}{\delta^6}\right)x^5\right) \, dx\\
        &\quad= \left[O\left(\frac{8(\alpha_G - \alpha_F)(\alpha_G + \alpha_F)\alpha_F}{\delta^3}\right)\frac{x^3}{3} 
        - O\left(\frac{16(\alpha_G - \alpha_F)(\alpha_G + \alpha_F)\alpha_F}{\delta^4}\right)\frac{x^4}{4}\right.\\
        &\quad\quad+\left. O\left(\frac{10(\alpha_G - \alpha_F)(\alpha_G + \alpha_F)\alpha_F}{\delta^5}\right)\frac{x^5}{5} - O\left(\frac{2(\alpha_G - \alpha_F)(\alpha_G + \alpha_F)\alpha_F}{\delta^6}\right)\frac{x^6}{6}\right]_0^\delta\\
        &\quad= \frac{8(\alpha_G - \alpha_F)(\alpha_G + \alpha_F)\alpha_F}{3} 
        - \frac{16(\alpha_G - \alpha_F)(\alpha_G + \alpha_F)\alpha_F}{4} + \frac{10(\alpha_G - \alpha_F)(\alpha_G + \alpha_F)\alpha_F}{5}\\
        &\quad\quad- \frac{2(\alpha_G - \alpha_F)(\alpha_G + \alpha_F)\alpha_F}{6} + O(\delta)\\
        &\quad= \frac{(\alpha_G - \alpha_F)(\alpha_G + \alpha_F)\alpha_F}{3} + O(\delta).
    \end{align*}
\end{proof}

We now upper-bound the gains-from-trade of our mechanism when $F = H_{\alpha_F, \delta}$ and $G = H_{\alpha_G, \delta},$ as $\delta \rightarrow 0,$ in the following lemma:
\begin{lemma}
\label{lem:gft-limit}
    Let $0 < \alpha_F \le \alpha_G < 1,$ Then, for any $\epsilon > 0,$ there exists $$0 < \delta < \min\left(\frac{\alpha_F}{2\alpha_F + 1}, \frac{1 - \alpha_F}{3 - 2\alpha_F}, \frac{\alpha_G}{2\alpha_G + 1}, \frac{1 - \alpha_G}{3 - 2\alpha_G}\right)$$ such that if $F = H_{\alpha_F, \delta}$ and $G = H_{\alpha_G, \delta},$ the gains-from-trade of the mechanism that draws samples $p \sim F, q \sim G,$ and offers price $p$ to the buyer and price $q$ to the seller achieves
    \begin{align*}
        \GFT < \frac{1}{12}\alpha_G(1 - \alpha_F)(\alpha_G - 2\alpha_F + \alpha_F\alpha_G + 2) + \epsilon.
    \end{align*}
\end{lemma}
\begin{proof}
    By Lemma~\ref{lem:sample-gft}, we have that
    \begin{align}
        \GFT &= \frac{1}{4}\int_0^1 G(x)^2(1 - F(x))^2 \, dx + \frac{1}{2}\int_0^1 G(q)^2 \int_q^1 (1-F(x)) \, dx \, dF(q)\nonumber\\
        \quad&+ \frac{1}{2}\int_0^1 (1 - F(p))^2 \int_0^p G(x) \, dx \, dG(p). \label{eq:01301001}
    \end{align}

    We now analyze each of the integrals in~\eqref{eq:01301001}.
    
    \paragraph{First Integral in~\eqref{eq:01301001}} We have that
    \begin{align*}
        \int_0^1 G(x)^2(1 - F(x))^2 \, dx &= \left(\int_0^\delta G(x)^2(1 - F(x))^2 \, dx + \int_\delta^{1 - \delta} G(x)^2(1 - F(x))^2 \, dx\right.\\
        &\qquad\qquad\left.+ \int_0^{1 - \delta} G(x)^2(1 - F(x))^2 \, dx\right)\\
        &\le \left(\delta + \int_\delta^{1 - \delta} G(x)^2(1 - F(x))^2 \, dx + \delta\right)\\
        &\le \left(\delta + \int_\delta^{1 - \delta} (\alpha_G + \delta)^2(1 - (\alpha_F - \delta))^2 \, dx + \delta\right)\\
        &= 2\delta + (1 - 2\delta)(\alpha_G + \delta)^2(1 - (\alpha_F - \delta))^2,
    \end{align*}
    where the last inequality holds since $F(x) \ge \alpha_F - \delta$ and $G(x) \le \alpha_G + \delta$ on $[\delta, 1 - \delta].$ 
    
    \paragraph{Second Integral in~\eqref{eq:01301001}} We have that
    \begin{align*}
        &\int_0^1 G(q)^2 \int_q^1 (1-F(x)) \, dx \, dF(q)\\ 
        &\quad= \int_0^1 F(q)^2 \int_q^1 (1-F(x)) \, dx \, dF(q) + \int_0^1 (G(q)^2 - F(q)^2) \int_q^1 (1-F(x)) \, dx \, dF(q)\\
        &\quad= \underbrace{\int_0^1 F(q)^2 \int_q^1 (1-F(x)) \, dx \, dF(q)}_{(a)} + \underbrace{\int_0^\delta (G(q)^2 - F(q)^2) \int_q^1 (1-F(x)) \, dx \, dF(q)}_{(b)}\\
        &\quad\quad+ \underbrace{\int_\delta^{1 - \delta} (G(q)^2 - F(q)^2) \int_q^1 (1-F(x)) \, dx \, dF(q)}_{(c)} + \underbrace{\int_{1 - \delta}^1 (G(q)^2 - F(q)^2) \int_q^1 (1-F(x)) \, dx \, dF(q)}_{(d)}.
    \end{align*}
    Now, integrating by parts as in the proof of Lemma~\ref{lem:sample-sym-gft}, we obtain that
    \begin{align*}
        (a) &= \frac{1}{3}\int_0^1 F(x)^3(1 - F(x)) \, dx\\
        &= \frac{1}{3}\left(\int_0^\delta F(x)^3(1 - F(x)) \, dx + \int_\delta^{1 - \delta} F(x)^3(1 - F(x)) \, dx + \int_{1 - \delta}^1 F(x)^3(1 - F(x)) \, dx\right)\\
        &\le \frac{1}{3}\left(\delta + \int_\delta^{1 - \delta} (\alpha_F + \delta)^3(1 - (\alpha_F - \delta)) \, dx + \delta\right)\\
        &= \frac{1}{3}(2\delta + (1 - 2\delta)(\alpha_F + \delta)^3(1 - (\alpha_F - \delta))).
    \end{align*}
    In addition,
    \begin{align*}
        (b) &= \int_0^\delta (G(q)^2 - F(q)^2)\left(\int_q^\delta (1 - F(x)) \, dx + \int_\delta^1 (1 - F(x)) \, dx\right) \, dF(q)\\
        &\le \int_0^\delta (G(q)^2 - F(q)^2)\left(\delta + \int_\delta^1 (1 - (\alpha_F - \delta)) \, dx\right) \, dF(q)\\
        &= \int_0^\delta (G(q)^2 - F(q)^2)(\delta + (1 - \delta)(1 - (\alpha_F - \delta))) \, dF(q)\\
        &= (\delta + (1 - \delta)(1 - (\alpha_F - \delta)))\int_0^\delta (G(q)^2 - F(q)^2) \, dF(q)\\
        &= (\delta + (1 - \delta)(1 - (\alpha_F - \delta)))\int_0^\delta (H_{\alpha_G, \delta, \QUAD}(q)^2 - H_{\alpha_F, \delta, \QUAD}(q)^2)H_{\alpha_F, \delta, \QUAD}'(q) \, dq\\
        &= (\delta + (1 - \delta)(1 - (\alpha_F - \delta)))\left(\frac{(\alpha_G - \alpha_F)(\alpha_G + \alpha_F)\alpha_F}{3} + O(\delta)\right),
    \end{align*}
    where the last equality holds by Lemma~\ref{lem:err-bnd}. Finally,
    \begin{align*}
        (c) \le \int_\delta^{1 - \delta} (G(q)^2 - F(q)^2) \, dF(q) \le \int_\delta^{1 - \delta} dF(q) = 2\delta
    \end{align*}
    and 
    \begin{align*}
        (d) &\le \int_{1 - \delta}^1 (G(q)^2 - F(q)^2) \int_{1 - \delta}^1 (1-F(x)) \, dx \, dF(q)\\
        &\le \delta\int_{1 - \delta}^1 (G(q)^2 - F(q)^2)\, dF(q)\\
        &\le \delta\int_{1 - \delta}^1 \, dF(q)\\
        &= \delta(1 - (\alpha_F + \delta)).
    \end{align*}
    Combining our results from above, we obtain that
    \begin{align*}
        \int_0^1 G(q)^2 \int_q^1 (1-F(x)) \, dx \, dF(q) &\le \frac{1}{3}(2\delta + (1 - 2\delta)(\alpha_F + \delta)^3(1 - (\alpha_F - \delta)))\\
        &\quad+ (\delta + (1 - \delta)(1 - (\alpha_F - \delta)))\left(\frac{(\alpha_G - \alpha_F)(\alpha_G + \alpha_F)\alpha_F}{3} + O(\delta)\right)\\
        &\quad+ 2\delta + \delta(1 - (\alpha_F + \delta)).
    \end{align*}

    \paragraph{Third Integral in~\eqref{eq:01301001}} We have that
    \begin{align*}
        &\int_0^1 (1 - F(p))^2 \int_0^p G(x) \, dx \, dG(p)\\
        &\quad= \int_0^1 (1 - G(p))^2 \int_0^p G(x) \, dx \, dG(p) + \int_0^1 ((1 - F(p))^2 - (1 - G(p))^2) \int_0^p G(x) \, dx \, dG(p)\\
        &\quad= \underbrace{\int_0^1 (1 - G(p))^2 \int_0^p G(x) \, dx \, dG(p)}_{(a)}\\
        &\quad\quad+ \underbrace{\int_0^\delta ((1 - F(p))^2 - (1 - G(p))^2) \int_0^p G(x) \, dx \, dG(p)}_{(b)}\\
        &\quad\quad+ \underbrace{\int_\delta^{1 - \delta} ((1 - F(p))^2 - (1 - G(p))^2) \int_0^p G(x) \, dx \, dG(p)}_{(c)}\\
        &\quad\quad+ \underbrace{\int_{1 - \delta}^1 ((1 - F(p))^2 - (1 - G(p))^2) \int_0^p G(x) \, dx \, dG(p)}_{(d)}.
    \end{align*}
    Now, integrating by parts as in the proof of Lemma~\ref{lem:sample-sym-gft}, we obtain that
    \begin{align*}
        (a) &= \frac{1}{3}\int_0^1 G(x)(1 - G(x))^3 \, dx\\
        &= \frac{1}{3}\left(\int_0^\delta G(x)(1 - G(x))^3 \, dx +\int_\delta^{1 - \delta} G(x)(1 - G(x))^3 \, dx + \int_{1 - \delta}^1 G(x)(1 - G(x))^3 \, dx\right)\\
        &\le \frac{1}{3}\left(\delta + \int_\delta^{1 - \delta} (\alpha_G + \delta)(1 - (\alpha_G - \delta))^3 \, dx + \delta\right)\\
        &= \frac{1}{3}(2\delta + (1 - 2\delta)(\alpha_G + \delta)(1 - (\alpha_G - \delta))^3).
    \end{align*}
    In addition, we have that
    \begin{align*}
        (d) &= \int_{1 - \delta}^1 ((1 - F(p))^2 - (1 - G(p))^2)\left(\int_0^{1 - \delta} G(x) \, dx + \int_{1 - \delta}^p G(x) \, dx\right) \, dG(p)\\
        &\le \int_{1 - \delta}^1 ((1 - F(p))^2 - (1 - G(p))^2)\left(\int_0^{1 - \delta} (\alpha_G + \delta) \, dx + \delta\right) \, dG(p)\\
        &= \int_{1 - \delta}^1 ((1 - F(p))^2 - (1 - G(p))^2)((1 - \delta)(\alpha_G + \delta) + \delta) \, dG(p)\\
        &= ((1 - \delta)(\alpha_G + \delta) + \delta)\int_{1 - \delta}^1 ((1 - F(p))^2 - (1 - G(p))^2) \, dG(p)\\
        &= ((1 - \delta)(\alpha_G + \delta) + \delta)\int_{1 - \delta}^1 (H_{1 - \alpha_F, \delta, \QUAD}(1 - p)^2 - H_{1 - \alpha_G, \delta, \QUAD}(1 - p)^2)H_{1 - \alpha_G, \delta, \QUAD}'(1 - p) \, dp\\
        &= ((1 - \delta)(\alpha_G + \delta) + \delta)\int_0^\delta (H_{1 - \alpha_F, \delta, \QUAD}(p)^2 - H_{1 - \alpha_G, \delta, \QUAD}(p)^2)H_{1 - \alpha_G, \delta, \QUAD}'(p) \, dp\\
        &= ((1 - \delta)(\alpha_G + \delta) + \delta)\left(\frac{(\alpha_G - \alpha_F)(2 - \alpha_F - \alpha_G)(1 - \alpha_G)}{3} + O(\delta)\right),
    \end{align*}
    where we perform a change of variables in the second-to-last equality and where the last equality holds by Lemma~\ref{lem:err-bnd}, with $1 - \alpha_G$ and $1 - \alpha_F$ taking the roles of $\alpha_F$ and $\alpha_G,$ respectively. Finally, we have that
    \begin{align*}
        (b) &\le \int_0^\delta ((1 - F(p))^2 - (1 - G(p))^2) \int_0^\delta G(x) \, dx \, dG(p)\\
        &\le \delta\int_0^\delta ((1 - F(p))^2 - (1 - G(p))^2) \, dG(p)\\
        &\le \delta\int_0^\delta dG(p)\\
        &= \delta(\alpha_G - \delta)
    \end{align*}
    and
    \begin{align*}
        (c) \le \int_\delta^{1 - \delta} ((1 - F(p))^2 - (1 - G(p))^2) \, dG(p) \le \int_\delta^{1 - \delta} dG(p) = 2\delta.
    \end{align*}
    Combining our results from above, we obtain that
    \begin{align*}
        \int_0^1 (1 - F(p))^2 \int_0^p G(x) \, dx \, dG(p) &\le \frac{1}{3}(2\delta + (1 - 2\delta)(\alpha_G + \delta)(1 - (\alpha_G - \delta))^3)\\
        &\quad+ ((1 - \delta)(\alpha_G + \delta) + \delta)\left(\frac{(\alpha_G - \alpha_F)(2 - \alpha_F - \alpha_G)(1 - \alpha_G)}{3} + O(\delta)\right)\\
        &\quad+ \delta(\alpha_G - \delta) + 2\delta.
    \end{align*}
    
    Finally, returning to~\eqref{eq:01301001}, we obtain that
    \begin{align*}
        \GFT &< \frac{\alpha_G^2(1 - \alpha_F)^2}{4} + \frac{\alpha_F^3(1 - \alpha_F)}{6} + \frac{(1 - \alpha_F)(\alpha_G - \alpha_F)(\alpha_G + \alpha_F)\alpha_F}{6}\\
        &\quad+ \frac{\alpha_G(1 - \alpha_G)^3}{6} + \frac{\alpha_G(\alpha_G - \alpha_F)(2 - \alpha_F - \alpha_G)(1 - \alpha_G)}{6} + \epsilon\\
        &= \frac{1}{12}\alpha_G(1 - \alpha_F)(\alpha_G - 2\alpha_F + \alpha_F\alpha_G + 2) + \epsilon
    \end{align*}
    for sufficiently small $\delta.$
\end{proof}

\section{Omitted Proofs}
\subsection{Proof of Lemma~\ref{lem:sample-gft}}
\begin{customlem}{\ref{lem:sample-gft}}
    The mechanism that draws samples $p \sim F, q \sim G$ and offers price $p$ to the buyer and price $q$ to the seller achieves
    \begin{align*}
         \Exu{p \sim F, q \sim G}{\GFT(p, q)} &= \frac{1}{4}\int_0^\infty G(x)^2(1 - F(x))^2 \, dx + \frac{1}{2}\int_0^\infty G(q)^2 \int_q^{\infty}(1-F(x)) \, dx \, dF(q)\\
        \quad&+ \frac{1}{2}\int_0^\infty (1 - F(p))^2 \int_0^p G(x) \, dx \, dG(p).
    \end{align*}
\end{customlem}
To prove Lemma~\ref{lem:sample-gft}, we will use the following result from the proof of Theorem 1.1 in~\citet{hajiaghayi2025bilateral}:
\begin{lemma}{\cite{hajiaghayi2025bilateral}}\label{lm:hajiaghayi}
    The gains-from-trade of the mechanism that offers price $p$ to the buyer and price $q$ to the seller, where $p \ge q,$ is given by
    \begin{align*}
        \GFT(p,q) &= \Pro(p, q) +  G(q)\int_{p}^\infty (1-F(x)) \, dx + (1-F(p))\int_{0}^q G(x) \, dx.
    \end{align*}
\end{lemma}

\begin{proof}[Proof of Lemma~\ref{lem:sample-gft}]
By Lemma~\ref{lm:hajiaghayi}, we have that
\begin{align*}
    \GFT(p,q) &= \Pro(p, q) +  G(q)\int_{p}^\infty (1-F(x)) \, dx + (1-F(p))\int_{0}^q G(x) \, dx
\end{align*}
for $p \ge q.$
Thus, together with Lemma~\ref{lem:sample-pft}, we have that
\begin{align*}
    \Exu{p \sim F, q \sim G}{\GFT(p, q)} &= \frac{1}{4}\int_0^\infty G(x)^2(1 - F(x))^2 \, dx + \lambda_1 + \lambda_2,
\end{align*}
where we define
\begin{align*}
    \lambda_1 &= \int_{0}^\infty \int_q^\infty G(q)\int_p^\infty (1-F(x)) \, dx \, dF(p) \, dG(q)\\
    \lambda_2 &= \int_0^\infty \int_q^\infty (1-F(p))\int_0^q G(x) \, dx \, dF(p) \, dG(q).
\end{align*}
Now, let $\displaystyle\lambda = \int_0^\infty (1-F(x)) \, dx.$ Then, we have that
\begin{align*}
    \lambda_1 &= \int_{0}^\infty G(q)\parans{\int_q^\infty \int_p^\infty (1-F(x)) \, dx \, dF(p)} \, dG(q)\\
    &= \int_{0}^\infty G(q)\parans{\int_q^\infty \left[\lambda - \int_0^p (1-F(x)) \, dx\right] dF(p)} \, dG(q)\\
    &= \int_{0}^\infty G(q)\parans{\lambda(1-F(q)) - \int_q^\infty \int_0^p (1-F(x)) \, dx \, dF(p)} \, dG(q).
\end{align*}
For the inner integral, notice that
\begin{align*}
    \int_q^\infty \int_0^p (1-F(x)) \, dx \, dF(p) &= F(p)\int_0^p (1-F(x)) \, dx \Bigg|_q^\infty - \int_q^\infty F(p)(1-F(p)) \, dp \tag{integrating by parts over $F'(p)$ and $\displaystyle\int_0^p(1-F(x)) \, dx$}
    \\
    &= \lambda - F(q)\int_0^q (1-F(x)) \, dx - \int_q^\infty F(p)(1-F(p)) \, dp.
\end{align*}
Hence, we can rewrite $\lambda_1$ as
\begin{align*}
    \lambda_1 &= \int_0^\infty G(q)\left(F(q)\int_0^q (1-F(x)) \, dx\right. + \left.\int_q^\infty F(p)(1 - F(p)) \, dp - \lambda F(q)\right) \, dG(q)
    \\
    &= 
    \int_0^\infty G(q) \parans{\int_q^{\infty}(1-F(x))(F(x) - F(q)) \, dx} \, dG(q)
    \\
    &= 
    \int_0^\infty G(q) \parans{\int_q^{\infty}(1-F(x))F(x) \, dx} \, dG(q) - \int_0^\infty F(q)G(q) \parans{\int_q^{\infty}(1-F(x)) \, dx} \, dG(q)
    \\
    &= 
    \int_0^\infty G(q) \parans{\int_q^{\infty}(1-F(x))F(x) \, dx} G'(q) \, dq - \int_0^\infty F(q)G(q) \parans{\int_q^{\infty}(1-F(x)) \, dx} G'(q) \, dq.
\end{align*}
Integrating by parts again, we obtain that
\begin{align*}
    \lambda_1 &= \left(\frac{G(q)^2}{2}\int_q^{\infty}(1-F(x))F(x) \, dx \Bigg|_0^\infty + \frac{1}{2}\int_0^\infty G(q)^2F(q)(1 - F(q)) \, dq\right)\\
    &- \left(\frac{G(q)^2}{2}F(q)\int_q^{\infty}(1-F(x)) \, dx \Bigg|_0^\infty\right. + \left.\frac{1}{2}\int_0^\infty G(q)^2\left[F(q)(1 - F(q)) - F'(q)\int_q^{\infty}(1-F(x)) \, dx\right] \, dq\right)\\
    &= \frac{1}{2}\int_0^\infty G(q)^2 \int_q^{\infty}(1-F(x)) \, dx \, dF(q).
\end{align*}
Finally, we have that
\begin{align*}
    \lambda_2 &= 
    \int_0^\infty \int_q^\infty \parans{(1-F(p))\int_0^q G(x) \, dx} \, dF(p) \, dG(q)
    \\
    &= \int_0^\infty \int_0^p \parans{(1-F(p)) \int_0^q G(x) \, dx} \, dG(q) \, dF(p)
    \\
    &= 
    \int_0^\infty (1-F(p)) \parans{\int_0^p \int_0^q G(x) \, dx \, dG(q)} \, dF(p)
    \\
    &= 
    \int_0^\infty (1-F(p)) \parans{\int_0^p \int_0^q G(x) \, dx \, G'(q) \, dq} \, dF(p).
\end{align*}
Integrating by parts over $G'(q)$ and $\displaystyle\int_0^q G(x) \, dx,$ we obtain that
\begin{align*}
    \lambda_2 
    &= \int_0^\infty (1-F(p))  
    \parans{G(q)\int_0^q G(x) \, dx \Bigg|_0^p - \int_0^p G(q)^2 \, dq} \, dF(p)
    \\
    &=
    \int_0^\infty (1-F(p)) \parans{\int_0^p G(x)(G(p) - G(x)) \, dx} \, dF(p)\\
    &=
    \int_0^\infty (1-F(p)) \parans{\int_0^p G(x)(1 - G(x)) \, dx} \, dF(p) - \int_0^\infty (1-F(p))(1 - G(p)) \parans{\int_0^p G(x) \, dx} \, dF(p)\\
    &= \int_0^\infty (1-F(p)) \parans{\int_0^p G(x)(1 - G(x)) \, dx} F'(p) \, dp - \int_0^\infty (1-F(p))(1 - G(p)) \parans{\int_0^p G(x) \, dx} F'(p) \, dp.
\end{align*}
Integrating by parts again, we obtain that
\begin{align*}
    \lambda_2 &= \left(-\frac{(1 - F(p))^2}{2}\int_0^p G(x)(1 - G(x)) \, dx \Bigg|_0^\infty\right. + \left.\frac{1}{2}\int_0^\infty (1 - F(p))^2(1 - G(p))G(p) \, dp\right)\\
    &-\left(-\frac{(1 - F(p))^2}{2}(1 - G(p))\int_0^p G(x) \, dx \Bigg|_0^\infty + \frac{1}{2}\int_0^\infty (1 - F(p))^2\left[(1 - G(p))G(p) - G'(p)\int_0^p G(x) \, dx\right] \, dp\right)\\
    &= \frac{1}{2}\int_0^\infty (1 - F(p))^2 \int_0^p G(x) \, dx \, dG(p).
\end{align*}
Finally, combining all our results from above, we conclude that
\begin{align*}
    \Exu{p \sim F, q \sim G}{\GFT(p, q)}
    &= \frac{1}{4}\int_0^\infty G(x)^2(1 - F(x))^2 \, dx + \frac{1}{2}\int_0^\infty G(q)^2 \int_q^{\infty}(1-F(x)) \, dx \, dF(q)\\
    \quad&+ \frac{1}{2}\int_0^\infty (1 - F(p))^2 \int_0^p G(x) \, dx \, dG(p).
\end{align*}
\end{proof}

\subsection{Proof of Theorem~\ref{thm:sample-const-apx-sw}}
\begin{customthm}{\ref{thm:sample-const-apx-sw}}
    In the symmetric setting, the mechanism that draws samples $p, q \sim F$ and offers price $\max(p, q)$ to the buyer and price $\min(p, q)$ to the seller achieves
    \begin{align*}
         \SW \ge \frac{2}{3}\FBW.
    \end{align*}
\end{customthm}
\begin{proof}
We have that
\begin{align*}
    \SW &= \Exu{c \sim F}{c} + \GFT
    \\
    &= 
    \int_0^\infty (1 - F(x)) \, dx + \left(\frac{1}{3}\int_0^\infty F(x)(1 - F(x)) \, dx\right. - \left.\frac{1}{6}\int_0^\infty F(x)^2(1 - F(x))^2 \, dx\right)
    \\
    &= \frac{1}{6}\left(\int_0^\infty (1 - F(x))(1 + F(x))(F(x)^2 - 2F(x) + 4) \, dx\right. +\left. 2\int_0^\infty (1 - F(x)) \, dx\right),
\end{align*}
where the second equality holds by Lemma~\ref{lem:sample-sym-gft}. 
Thus, since $F(x)^2 - 2F(x) + 4 = (F(x) - 1)^2 + 3 \ge 3$ and $1 + F(x) \le 2$ for all $x \ge 0,$
we have that
\begin{align*}
    \SW &= \frac{1}{6}\left(\int_0^\infty (1 - F(x))(1 + F(x))(F(x)^2 - 2F(x) + 4) \, dx\right. + \left.2\int_0^\infty (1 - F(x)) \, dx\right)\\
    &\ge \frac{1}{6}\left(3\int_0^\infty (1 - F(x))(1 + F(x)) \, dx\right. + \left.\int_0^\infty (1 - F(x))(1 + F(x)) \, dx\right)\\
    &= \frac{2}{3}\int_0^\infty (1 - F(x))(1 + F(x)) \, dx\\
    &= \frac{2}{3}\FBW.
\end{align*}
\end{proof}

\subsection{Proof of Theorem~\ref{thm:sample-asym-const-sw}}
\begin{customthm}{\ref{thm:sample-asym-const-sw}}
    Suppose $F$ stochastically dominates $G.$ Then, the mechanism that draws samples $p \sim F, q \sim G$ and offers price $p$ to the buyer and price $q$ to the seller achieves
    \begin{align*}
         \SW \ge \frac{3 - \sqrt{2}}{12}\FBW,
    \end{align*}
    where $(3 - \sqrt{2}) / 12 \approx 0.1321.$
\end{customthm}
\begin{proof}
By Lemma~\ref{lem:sample-stoc-dom-gft}, 
\begin{align*}
    \SW &= \Exu{c \sim G}{c} + \GFT\\ 
    &\ge \int_0^\infty (1 - G(x)) \, dx + \frac{1}{12}\int_0^\infty (3G(x)^2(1 - F(x))^2 + 2F(x)^3(1 - F(x)) + 2(1 - G(x))^3G(x)) \, dx.
\end{align*}
Now, consider the following optimization problem:
\begin{equation*}
    \begin{aligned}
        & \underset{f, g}{\text{minimize}}
        & & \frac{12(1 - g) + 3g^2(1 - f)^2 + 2f^3(1 - f) + 2(1 - g)^3g}{12(1 - fg)} \\
        & \text{subject to}
        & & 0 \le f \le g \le 1.
    \end{aligned}
\end{equation*}
Solving this numerically,\footref{ft:wolfram} we obtain that the optimum $\alpha = (3 - \sqrt{2}) / 12$ occurs at $(f, g) = (1 / \sqrt{2}, 1).$ Thus,
\begin{align*}
    &(1 - G(x)) + \frac{1}{12}(3G(x)^2(1 - F(x))^2 + 2F(x)^3(1 - F(x)) + 2(1 - G(x))^3G(x))\\
    &\quad\ge \frac{3 - \sqrt{2}}{12}((1 - G(x)) + G(x)(1 - F(x)))
\end{align*}
for all $x \ge 0,$ so 
\begin{align*}
    \SW &\ge \frac{3 - \sqrt{2}}{12}\parans{\int_0^\infty (1 - G(x)) \, dx + \int_0^\infty G(x)(1 - F(x)) \, dx} = \frac{3 - \sqrt{2}}{12}\FBW.
\end{align*}
\end{proof}

\subsection{Proof of Theorem~\ref{thm:asym-ub}}
Similar to the symmetric setting, our proof of the lower bound in Theorem~\ref{thm:sample-asym-const-gft} suggests that we can prove a matching upper bound for our mechanism by taking $F(x) \approx \alpha_F$ and $G(x) \approx \alpha_G,$ where $(\alpha_F, \alpha_G)$ corresponds to the solution of the optimization problem in the proof of Theorem~\ref{thm:sample-asym-const-gft}. Unfortunately, since the proof of Theorem~\ref{thm:sample-asym-const-gft} uses the lower bound of Lemma~\ref{lem:sample-stoc-dom-gft}, which may not be tight, choosing these values for $\alpha_F$ and $\alpha_G$ will not necessarily give us the best upper bound. Thus, we instead consider $F = H_{\alpha_F, \delta}$ and $G = H_{\alpha_G, \delta}$ for $0<\alpha_F \le \alpha_G < 1$ and $\delta > 0$ \textit{to be determined} and optimize over $\alpha_F, \alpha_G,$ and $\delta$ to obtain the following upper bound on the approximation ratio of our mechanism to the first-best gains-from-trade:
\begin{theorem}
\label{thm:ub-gft-asym}
    For any $\epsilon > 0,$ there exists $F$ and $G,$ where $F$ stochastically dominates $G,$ such that the mechanism that draws samples $p \sim F, q \sim G$ and offers price $p$ to the buyer and price $q$ to the seller achieves
    \begin{align*}
        \GFT < \left(\frac{7}{48} + \epsilon\right)\FB.
    \end{align*}
\end{theorem}
\begin{proof}
    Let $0 < \alpha_F \le \alpha_G < 1$ and $$0 < \delta <\min\left(\frac{\alpha_F}{2\alpha_F + 1}, \frac{1 - \alpha_F}{3 - 2\alpha_F}, \frac{\alpha_G}{2\alpha_G + 1}, \frac{1 - \alpha_G}{3 - 2\alpha_G}\right)$$ be constants to be determined later. Now, let $F = H_{\alpha_F, \delta}$ and $G = H_{\alpha_G, \delta}.$ Note that $F$ stochastically dominates $G,$ by Lemma~\ref{lem:stoc-dom}. 
    
    Now, let $\epsilon > 0.$ We have that
    \begin{align*}
        \FB &= \int_0^1 G(x)(1 - F(x)) \, dx\\
        &\ge \int_\delta^{1 - \delta} G(x)(1 - F(x)) \, dx\\
        &\ge \int_\delta^{1 - \delta} (\alpha_G - \delta)(1 - (\alpha_F + \delta)) \, dx\\
        &= (1 - 2\delta)(\alpha_G - \delta)(1 - (\alpha_F + \delta)),
    \end{align*}
    so by Lemma~\ref{lem:gft-limit},
    \begin{align*}
        \frac{\GFT}{\FB} < \frac{1}{12}(\alpha_G - 2\alpha_F + \alpha_F\alpha_G + 2) + \epsilon 
    \end{align*}
    for sufficiently small $\delta.$ Finally, since the ratio above is minimized when $\alpha_F = \alpha_G$ and $\alpha_G = 1 / 2,$ we obtain that
    \begin{align*}
        \frac{\GFT}{\FB} < \frac{7}{48} + \epsilon
    \end{align*}
    for sufficiently small $\delta.$
\end{proof}

Likewise, we obtain the following upper bound on the approximation ratio of our mechanism to the first-best social welfare: 
\begin{theorem}
\label{thm:ub-sw-asym}
    For any $\epsilon > 0,$ there exists $F$ and $G,$ where $F$ stochastically dominates $G,$ such that the mechanism that draws samples $p \sim F, q \sim G$ and offers price $p$ to the buyer and price $q$ to the seller achieves
    \begin{align*}
        \SW < \left(\frac{1}{6} + \epsilon\right)\FBW.
    \end{align*}
\end{theorem}
\begin{proof}
    Our proof follows a similar structure as the proof of Theorem~\ref{thm:ub-gft-asym}. Let $0 < \alpha_F \le \alpha_G < 1$ and $$0 < \delta <\min\left(\frac{\alpha_F}{2\alpha_F + 1}, \frac{1 - \alpha_F}{3 - 2\alpha_F}, \frac{\alpha_G}{2\alpha_G + 1}, \frac{1 - \alpha_G}{3 - 2\alpha_G}\right)$$ be constants to be determined later. Now, let $F = H_{\alpha_F, \delta}$ and $G = H_{\alpha_G, \delta}.$ Note that $F$ stochastically dominates $G,$ by Lemma~\ref{lem:stoc-dom}. 
    
    Now, let $\epsilon' > 0.$ We have that
    \begin{align*}
        \Exu{c \sim G}{c} &= \int_0^1 (1 - G(x)) \, dx\\
        &\le 2\delta + \int_\delta^{1 - \delta} (1 - G(x)) \, dx\\
        &\le 2\delta + \int_\delta^{1 - \delta} (1 - (\alpha_G - \delta)) \, dx\\
        &= 2\delta + (1 - 2\delta)(1 - (\alpha_G - \delta))
    \end{align*}
    and
    \begin{align*}
        \FBW &= \Exu{c \sim G}{c} + \FB\\
        &= \int_0^1 (1 - G(x)) \, dx + \int_0^1 G(x)(1 - F(x)) \, dx\\
        &\ge \int_\delta^{1 - \delta} (1 - G(x)) \, dx + \int_\delta^{1 - \delta} G(x)(1 - F(x)) \, dx\\
        &\ge \int_\delta^{1 - \delta} (1 - (\alpha_G + \delta)) \, dx + \int_\delta^{1 - \delta} (\alpha_G - \delta)(1 - (\alpha_F + \delta)) \, dx\\
        &= (1 - 2\delta)(1 - (\alpha_G - \delta)) + (1 - 2\delta)(\alpha_G - \delta)(1 - (\alpha_F + \delta)).
     \end{align*}
    Thus, since $\SW = \GFT + \Exu{c \sim G}{c},$ we have that
    \begin{align*}
        \frac{\SW}{\FBW} &< \frac{12(1 - \alpha_G) + \alpha_G(1 - \alpha_F)(\alpha_G - 2\alpha_F + \alpha_F\alpha_G + 2)}{12(1 - \alpha_F\alpha_G)} + \epsilon'\\
        &= \frac{12 - \alpha_F^2\alpha_G^2 + 2\alpha_F^2\alpha_G - 4\alpha_F\alpha_G + \alpha_G^2 - 10\alpha_G}{12(1 - \alpha_F\alpha_G)} + \epsilon'
    \end{align*}
    for sufficiently small $\delta,$ by Lemma~\ref{lem:gft-limit}. 

    It remains to minimize the ratio above subject to $0 < \alpha_F \le \alpha_G < 1.$ Taking the derivative with respect to $\alpha_G,$ we obtain that
    \begin{align*}
        &\frac{\partial}{\partial \alpha_G}\left[\frac{12 - \alpha_F^2\alpha_G^2 + 2\alpha_F^2\alpha_G - 4\alpha_F\alpha_G + \alpha_G^2 - 10\alpha_G}{12(1 - \alpha_F\alpha_G)}\right]\\
        &\quad= \frac{(1 - \alpha_F\alpha_G)(-2\alpha_F^2\alpha_G + 2\alpha_F^2 - 4\alpha_F + 2\alpha_G - 10)}{12(1 - \alpha_F\alpha_G)^2}\\
        &\quad+ \frac{\alpha_F(12 - \alpha_F^2\alpha_G^2 + 2\alpha_F^2\alpha_G - 4\alpha_F\alpha_G + \alpha_G^2 - 10\alpha_G)}{12(1 - \alpha_F\alpha_G)^2}\\
        &\quad= \frac{(1 - \alpha_F)(-\alpha_F^2\alpha_G^2 - \alpha_F\alpha_G^2 + 2\alpha_F\alpha_G - 2\alpha_F + 2\alpha_G - 10)}{12(1 - \alpha_F\alpha_G)^2}\\
        &\quad\le \frac{(1 - \alpha_F)(-\alpha_F^2\alpha_G^2 - \alpha_F\alpha_G^2 - 2\alpha_F - 6)}{12(1 - \alpha_F\alpha_G)^2}\le 0,
    \end{align*}
    so for any fixed $\alpha_F,$ the ratio above is minimized as $\alpha_G \rightarrow 1.$ Hence, choosing $\alpha_F$ that minimizes
    \begin{align*}
        \lim_{\alpha_G \rightarrow 1} \frac{12 - \alpha_F^2\alpha_G^2 + 2\alpha_F^2\alpha_G - 4\alpha_F\alpha_G + \alpha_G^2 - 10\alpha_G}{12(1 - \alpha_F\alpha_G)} &= \frac{3 + \alpha_F^2 - 4\alpha_F}{12(1 - \alpha_F)} = \frac{3 - \alpha_F}{12},
    \end{align*}
    we take $\alpha_F \rightarrow 1.$ 
    Therefore, $\forall \epsilon > 0,$ we have that $\displaystyle \frac{\SW}{\FBW} < \frac{1}{6} + \epsilon$
    for sufficiently small $\delta$ and suitable choices of $0 < \alpha_F \le \alpha_G < 1.$ 
\end{proof}
\end{document}